%% file: main.tex
\documentclass[journal]{IEEEtran}
\usepackage{amsthm}
\usepackage{amsmath,amsfonts}
\usepackage{algorithmic}
\usepackage{algorithm}
\usepackage{array}
\usepackage[caption=false,font=normalsize,labelfont=sf,textfont=sf]{subfig}
\usepackage{textcomp}
\usepackage{amssymb}
\usepackage{stfloats}
\usepackage{url}
\usepackage{verbatim}
\usepackage{graphicx}
\usepackage{cite}
\usepackage{acronym}
\usepackage{orcidlink}
\usepackage{lipsum}
\usepackage{cancel}

\usepackage{color,soul}
\PassOptionsToPackage{bookmarks=false}{hyperref}
\hypersetup{hidelinks}
\input{acronyms.tex}
\usepackage[normalem]{ulem}
\usepackage[user]{zref}
\newcounter{revc}
\makeatletter \zref@newprop{revcontent}{} \zref@addprop{main}{revcontent}
\zref@newprop{revsec}{} \zref@addprop{main}{revsec}
\zref@newprop{revpage}{} \zref@addprop{main}{revpage}

\newcommand{\revi}[2]{
\zref@setcurrent{revsec}{\thesection}%
\zref@setcurrent{revpage}{\thepage}%
\zref@setcurrent{revcontent}{#2}%
\refstepcounter{revc}%
\label{#1}%
\zlabel{#1}%
#2%
}

\newcommand{\revinu}[2]{%
\zref@setcurrent{revsec}{\thesection}%
\zref@setcurrent{revcontent}{#2}%
\refstepcounter{revc}%
\zlabel{#1}%
\label{#1}
#2 }

\newcommand{\revr}[2]{%
\zref@setcurrent{revsec}{\thesection}%
\zref@setcurrent{revcontent}{#2}%
\refstepcounter{revc}%
\zlabel{#1}%
\label{#1} \sot{#2}} \makeatother

\usepackage[framemethod=tikz]{mdframed}
\definecolor{mycolor}{rgb}{0.122, 0.435, 0.698}
\newmdenv[innerlinewidth=0.5pt, roundcorner=4pt,linecolor=mycolor,innerleftmargin=6pt,innerrightmargin=6pt,innertopmargin=6pt,innerbottommargin=6pt]{mybox}

\begin{document}

\title{Extremely Large Beyond-Diagonal RIS: Low-Rank Modal Optimization for Near-Field Communications}
\author{Giovanni Iacovelli$^{\orcidlink{0000-0002-3551-4584}}$,~\IEEEmembership{Member,~IEEE}, Chandan~Kumar~Sheemar$^{\orcidlink{0000-0003-1676-5983}}$,~\IEEEmembership{Member,~IEEE}, Eva Lagunas$^{\orcidlink{0000-0002-9936-7245}}$,~\IEEEmembership{Senior Member,~IEEE}, and Symeon Chatzinotas$^{\orcidlink{0000-0001-5122-0001}}$,~\IEEEmembership{Fellow,~IEEE}
\thanks{Copyright (c) 2025 IEEE. Personal use of this material is permitted. However, permission to use this material for any other purposes must be obtained from the IEEE by sending a request to pubs-permissions@ieee.org. G. Iacovelli, C.K. Sheemar, E. Lagunas, and S. Chatzinotas are with the Signal Processing and Communications (SIGCOM) Research Group at Interdisciplinary Centre for Security, Reliability and Trust (SnT), University of Luxembourg, 1855 Luxembourg City, Luxembourg (emails: giovanni.iacovelli@uni.lu, chandankumar.sheemar@uni.lu, eva.lagunas@uni.lu, symeon.chatzinotas@uni.lu).}}

\newtheorem{assumption}{Assumption}
\newtheorem{remark}{Remark}
\newtheorem{theorem}{Theorem}
\newtheorem{corollary}{Corollary}
\newtheorem{lemma}{Lemma}
\newtheorem{proposition}{Proposition}

\maketitle
\begin{abstract}
Beyond-diagonal reconfigurable intelligent surfaces (BD-RIS) achieve their best performance when fully connected, at the price of an optimization and hardware burden that grows quadratically, and per iteration cubically, with the number of elements. Extremely large surfaces make this burden prohibitive, while their sheer aperture places both the base station and the users in the radiative near field, where far-field design tools break down. This paper introduces the extremely large BD-RIS (XL-BD-RIS) concept and shows that near-field geometry is precisely what makes fully connected performance affordable at scale. Modeling the cascade with the free-space Green function, we prove that the aperture fields live in a low-dimensional subspace spanned by the spherical-wave responses of the terminal positions, and we design a compact unitary modal matrix on this subspace, built from localization information alone, that provably attains the fully connected optimum with a number of reconfigurable entries set by the geometry and independent of the panel size. A weighted-MMSE Riemannian algorithm optimizes the beamformers and the modal matrix with monotone convergence at panel-size-independent cost. Numerical results show that a $24\times24$-element panel reaches the fully connected optimum with about two hundred entries instead of three hundred thousand. A mismatched DFT beamspace pays a sixty-fold entry penalty rooted in the beam spread of spherical wavefronts, while the classical block-wise architecture delivers strictly lower rates at any matched entry budget.
\end{abstract}

\section{Introduction}
\allowdisplaybreaks
 \IEEEPARstart{B}{eyond-diagonal} reconfigurable intelligent surfaces (BD-RIS) have emerged as an advanced RIS architecture that overcomes the limitations imposed by element-wise independent control \cite{shen2022modeling,li2022beyond}. By interconnecting the surface elements through a reconfigurable impedance network, BD-RIS enables non-diagonal scattering responses that, in the lossless reciprocal case, are unitary and symmetric, thereby exploiting a substantially larger scattering-parameter space and enabling more general wave transformations than conventional diagonal RIS architectures. The BD-RIS framework encompasses single-, group-, and fully-connected architectures with different complexity-performance trade-offs \cite{li2022beyond}, multi-sector configurations for full-space coverage \cite{li2023multisector}, closed-form and globally optimal designs for single-user links \cite{nerini2024closedform}, and graph-theoretic formulations identifying minimal circuit topologies capable of retaining optimal performance \cite{nerini2024graph}. Comprehensive tutorial and survey treatments are provided in \cite{li2025tutorialbeyonddiagonal,khan2025survey}. Overall, fully-connected BD-RIS offers the greatest wave-manipulation flexibility and typically achieves the highest performance, whereas group-connected architectures provide a practical trade-off between performance and circuit complexity.

In parallel, the push toward higher carriers and larger apertures has moved practical deployments into the radiative near field. For an extremely large (XL) surface of aperture $D$, operated at the millimeter-wave and sub-THz carriers where such apertures are envisioned, the Rayleigh distance $2D^{2}/\lambda$ reaches tens or hundreds of meters (about $3$~m already for a $13$~cm panel at $28$~GHz), so that both the base station (BS) and the users naturally reside where wavefronts are spherical rather than planar \cite{dardari2020communicating,cuidai2022,zhang2022beamfocusing}. This regime is not a nuisance to be corrected but a resource. Spherical wavefronts carry range information, enable beam focusing rather than mere beam steering, and enlarge the spatial degrees of freedom beyond the far-field angular picture \cite{miller2000,pizzo2020spatially,pizzo2022fourier,zhang2022beamfocusing}. At the same time, the near field invalidates the modeling shortcuts on which much of the RIS literature rests. Planar-wavefront steering vectors, angle-only channel parameterizations, and Fourier (DFT) beamspace representations all degrade as the geometry deepens \cite{cuidai2022}.

These two trends collide in the optimization. The natural variable of a fully-connected BD-RIS is an $M\times M$ unitary (symmetric) scattering matrix, and state-of-the-art designs optimize it in the element domain, with per-iteration costs growing as $O(M^{3})$ and a quadratic number $O(M^{2})$ of reconfigurable circuit entries \cite{li2022beyond,li2025tutorialbeyonddiagonal,nerini2024graph}. For the surface sizes that make the near field relevant in the first place, hundreds to thousands of elements, both scalings are prohibitive. The algorithms do not converge in reasonable time, and the hardware interface does not scale. The group-connected family reduces both costs, but abandons any optimality guarantee. Its feasible set is a measure-zero subset of the unitary group, and how much of the diagonal-to-fully-connected gap a given group size recovers is an empirical property of the channel. Conversely, dimensionality reduction through a fixed generic basis, most prominently the DFT beamspace inherited from far-field massive MIMO, re-introduces the very far-field assumption the regime violates. As we quantify in this paper, a spherical wavefront atom occupies on the order of $(d_{\mathrm F}/2r)^{2}$ beams at range $r$, so a beamspace-restricted design pays an entry-count penalty that grows quadratically with near-field depth. What is missing is a dimensionality reduction that is matched to the near-field geometry and comes with a provable equivalence to the fully-connected optimum. This calls for an architecture whose reconfigurable dimension is set by the propagation environment, namely the number of coupled spherical-wave modes, rather than by the panel size, together with an optimization scheme whose per-iteration cost is independent of $M$. Providing exactly this is the purpose of the present paper. Our design philosophy is related in spirit to the channel-operator diagonalization of \cite{iacovelli2025hmiso} for holographic MISO links, here generalized from a single active aperture to the two-leg reflective cascade of a BD-RIS. In the present setting the design variable is the surface response rather than the transmit current, and unitarity of the compressed response must be preserved exactly.

The main contributions are as follows.
\begin{itemize}
\item \textbf{\emph{Geometric modal architecture.}} Modeling both legs of the cascade with a single free-space Green function, we show that all fields on the aperture live in the span of the $N{+}K$ position-determined spherical-wave atoms, and we construct a modal basis from one thin SVD of the stacked atom matrix, completed by an arbitrary orthonormal slack of equal dimension. Performance depends on the basis only through its column space (Prop.~\ref{prop:subspace}), so the construction requires localization information alone.
\item \textbf{\emph{Exact $M$-free equivalence.}} Via a unitary dilation argument (Prop.~\ref{prop:halmos}), an $L\times L$ unitary modal matrix with $L=2\rho\le2(N{+}K)$, embedded in the surface through a lossless completion, is provably equivalent to an unconstrained $M\times M$ unitary surface: the fully-connected optimum is attained with a number of reconfigurable entries $L^{2}$ set by the propagation geometry and independent of the panel size. The dimension $\rho$ is the numerical rank of the stacked atom matrix, measured rather than assumed, and the rule adapts as the geometry changes. The equivalence itself relies only on the subspace structure of the channel and holds in any \ac{LoS} geometry. The near field is where it is both constructible, since the subspace follows from terminal positions through the spherical-wave atoms, and valuable, since in the far field the dimension collapses toward $K{+}1$ and a plane-wave beamspace already spans the subspace.
\item \textbf{\emph{Scalable optimization with guarantees.}} We develop a WMMSE-Riemannian alternation for the weighted sum-rate problem: closed-form beamforming with an exact radial power characterization (Lemma~\ref{lem:radial}), and a Riemannian update of the modal matrix on the unitary group with a low-rank line search whose per-iteration cost is $O(L^{3})$, independent of $M$, with monotone convergence of the objective (Prop.~\ref{prop:monotone}). The same alternating scheme specializes to the entire element-domain family, with exact coordinate descent for the diagonal surface, blockwise Riemannian updates for group-connected networks, and the fully-connected anchor, yielding a complete and mutually consistent complexity-performance ladder in which the modal architecture is the only provably anchor-attaining entry.
\item \textbf{\emph{Near-field analysis and validation.}} Extensive experiments in exact spherical-wave geometry quantify three effects. The first is the entry-count penalty of the mismatched DFT beamspace, with its closed-form origin in the $(d_{\mathrm F}/2r)^{2}$ beam spread of spherical atoms. The second is near-field range resolvability, i.e., two users at the same azimuth separated only in range, including the alignment of the rate transition with the depth-of-focus law $2r_0^{2}/d_{\mathrm F}$. The third is the dependence of the dimensioning rule and of all architectures on near-field depth, showing in particular that the advantage of beyond-diagonal coupling over diagonal phasing is itself an increasing function of depth. All optimality claims are certified numerically by cross-architecture exchanges and fully-connected checks.
\end{itemize}

Notation: Lowercase and uppercase boldface denote vectors and matrices. $(\cdot)^{\mathsf T}$, $(\cdot)^{*}$, and $(\cdot)^{\mathsf H}$ denote transpose, conjugate, and conjugate transpose. $\mathrm{tr}(\cdot)$, $\|\cdot\|$, and $\|\cdot\|_{\mathrm F}$ are the trace, the Euclidean norm, and the Frobenius norm. $\mathbf{I}_n$ is the $n\times n$ identity and $\mathcal{U}(n)$ the group of $n\times n$ unitary matrices. $\mathrm{blkdiag}(\cdot)$ is the block-diagonal operator, $\mathbb{E}[\cdot]$ denotes expectation, and $\jmath=\sqrt{-1}$. The span of the columns of $\mathbf{A}$ is $\mathrm{span}(\mathbf{A})$, and $\mathrm{rank}_\epsilon(\mathbf{A})$ denotes the numerical rank at relative threshold $\epsilon$.

The remainder of the paper is organized as follows. Sec.~\ref{sec:system_model} introduces the continuous-aperture near-field system model. Sec.~\ref{sec:modal} develops the modal representation of the cascaded channel and its lossless completion. Sec.~\ref{sec:basis} constructs the geometric modal basis and establishes the exactness result. Sec.~\ref{sec:sumrate} presents the WMMSE-Riemannian sum-rate maximization and its complexity. Sec.~\ref{sec:architectures} extends the scheme to the element-domain architecture family. Sec.~\ref{sec:numerical} reports the numerical study, and Sec.~\ref{sec:conclusion} concludes.

\section{System Model}
\label{sec:system_model}
We consider the narrowband downlink depicted in Fig.~\ref{fig:scenario}, in which a multi-antenna \ac{BS} serves $K$ single-antenna users with the help of an XL-BD-RIS. The \ac{BS} is equipped with $N\ge K$ antennas at positions $\{\mathbf{b}_k\}_{k=1}^{N}$ and transmits one data stream per user. The direct \ac{BS}-user links are assumed to be blocked, so that communication is established solely through the cascaded \ac{BS}-RIS-user path. This is a standard modeling choice that isolates the BD-RIS contribution and can be relaxed by superimposing a residual direct term.

The \ac{BS} sends the symbol vector $\mathbf{s}=[s_1,\dots,s_K]^\top$ with $\mathbb{E}[\mathbf{s}\mathbf{s}^\mathsf{H}]=\mathbf{I}_K$, applying a beamformer $\mathbf{w}_k\in\mathbb{C}^{N}$ to each stream under the total power budget $\sum_{k}\|\mathbf{w}_k\|^2\le P$.

\begin{figure}[t]
\centering
\includegraphics[width=1\columnwidth]{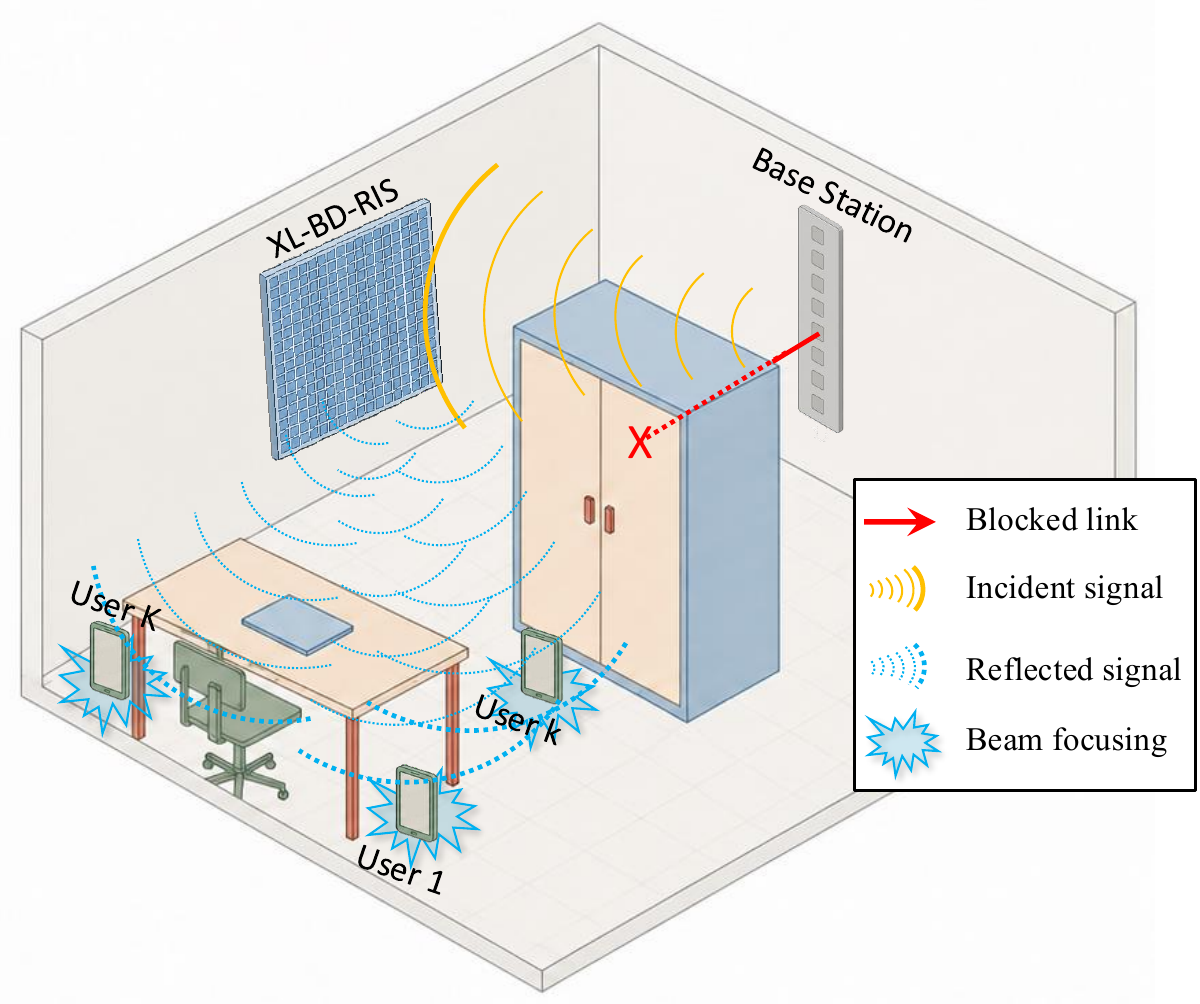}
\caption{Reference scenario. An XL-BD-RIS embedded in a wall serves $K$ ground users whose direct \ac{BS} links are blocked. Both legs of the cascade exhibit spherical wavefronts, and the surface focuses, rather than steers, the reflected field toward each user (drawing not to scale).}
\label{fig:scenario}
\end{figure}

\subsection{Continuous-Aperture Channel}
\label{sec:cont_channel}
To expose the spatial structure of the surface before discretization, we first describe the XL-BD-RIS as a continuous aperture $\mathcal{A}\subset\mathbb{R}^3$, and we let $\mathbf{r}\in\mathcal{A}$ denote a generic point on it. The continuous description is an analytical device: it lets us characterize the field over the surface independently of any particular element layout, while the physical surface will be recovered in Sec.~\ref{sec:discretization} by sampling $\mathcal{A}$ at $M$ elements on a $\lambda/2$ grid. The description is accurate whenever the aperture is electrically large, i.e., its extent greatly exceeds $\lambda$, so that the field varies smoothly across it.

Both legs of the cascade are described by one and the same propagation primitive: the free-space scalar Green function between two points $\mathbf{x},\mathbf{y}\in\mathbb{R}^3$,
\begin{equation}
\mathcal{G}(\mathbf{x},\mathbf{y})=\frac{e^{-j\kappa_0\|\mathbf{x}-\mathbf{y}\|}}{\|\mathbf{x}-\mathbf{y}\|},
\label{eq:green_function}
\end{equation}
with wavenumber $\kappa_0=2\pi/\lambda$, which retains both the spherical phase curvature and the amplitude decay of a point source. It is symmetric in its arguments, $\mathcal{G}(\mathbf{x},\mathbf{y})=\mathcal{G}(\mathbf{y},\mathbf{x})$, reflecting propagation reciprocity. We adopt a \ac{LoS}-dominant model on both legs, consistent with deployment practice: the surface is installed precisely so that it maintains visual links to the \ac{BS} and to the served area, while the direct \ac{BS}-user paths are blocked. Each leg is then a sampled Green function: the \ac{BS}-RIS response from antenna $k$ to an aperture point $\mathbf{r}$ is
\begin{equation}
g_k(\mathbf{r})=\sqrt{G_{\mathrm{BS}}}\,\mathcal{G}\!\left(\mathbf{r},\mathbf{b}_k\right),\qquad k=1,\dots,N,
\label{eq:bs_ris_scattered_channel}
\end{equation}
where $G_{\mathrm{BS}}$ collects the \ac{BS} antenna gain, and the RIS-user-$n$ response, with the user at $\mathbf{u}_n$ and $\mathbf{r}$ the aperture point acting as the source, is
\begin{equation}
f_n(\mathbf{r})=\mathcal{G}\!\left(\mathbf{u}_n,\mathbf{r}\right)=\mathcal{G}\!\left(\mathbf{r},\mathbf{u}_n\right),\qquad n=1,\dots,K .
\label{eq:ris_user_scattered_channel}
\end{equation}
The entire propagation environment is thus generated by the $N{+}K$ terminal positions $\{\mathbf{b}_k\}\cup\{\mathbf{u}_n\}$ through the single kernel \eqref{eq:green_function}. There are no random path gains, and every quantity derived below is a deterministic function of the deployment geometry. Two consequences are worth recording at the outset. First, across an electrically large aperture the factor $1/\|\mathbf{r}-\mathbf{b}_k\|$ imprints a smooth but non-negligible amplitude taper on each leg, in addition to the spherical phase: it is exactly this amplitude structure, which a phase-only surface cannot reshape, that will separate the reconfigurable architectures compared in this paper. Second, since every length in the model (element pitch, aperture, ranges) will be specified in units of $\lambda$, the geometry, the channel matrices, and hence all rates are invariant under a common rescaling of the wavelength: results stated at one carrier transfer verbatim to any other, only the physical size of the deployment changes.

Stacking the per-antenna responses into $\mathbf{g}(\mathbf{r})=[g_1(\mathbf{r}),\dots,g_N(\mathbf{r})]\in\mathbb{C}^{1\times N}$, the field impinging on the aperture under beamformer $\mathbf{w}_k$ is
\begin{equation}
e_k(\mathbf{r})=\mathbf{g}(\mathbf{r})\,\mathbf{w}_k .
\label{eq:incident_field}
\end{equation}
The XL-BD-RIS acts on the incident field through a linear scattering kernel $\theta(\mathbf{r},\mathbf{r}')$, producing the surface field
\begin{equation}
u_k(\mathbf{r})=\int_{\mathcal{A}}\theta(\mathbf{r},\mathbf{r}')\,e_k(\mathbf{r}')\,d\mathbf{r}',
\label{eq:continuous_scattering}
\end{equation}
which, for a passive lossless surface, corresponds to a unitary kernel\footnote{We retain the general unitary model, which is physically realizable through non-reciprocal circuit elements such as circulators \cite{bjornson2025capacity}, the symmetric constraint is left as an extension.}
\begin{equation}
\int_{\mathcal{A}}\theta(\mathbf{r},\mathbf{r}')\,\theta^{*}(\mathbf{r}'',\mathbf{r}')\,d\mathbf{r}'=\delta(\mathbf{r}-\mathbf{r}'').
\label{eq:continuous_unitarity}
\end{equation}
Radiating the surface field toward user $n$ and combining \eqref{eq:continuous_scattering} gives the end-to-end gain
\begin{equation}
h_{n,k}=\int_{\mathcal{A}}f_n(\mathbf{r})\,u_k(\mathbf{r})\,d\mathbf{r}
=\!\iint_{\mathcal{A}\times\mathcal{A}}\!f_n(\mathbf{r})\,\theta(\mathbf{r},\mathbf{r}')\,e_k(\mathbf{r}')\,d\mathbf{r}\,d\mathbf{r}' .
\label{eq:full_cascaded_continuous}
\end{equation}
Because $e_k(\mathbf{r}')=\mathbf{g}(\mathbf{r}')\mathbf{w}_k$ is linear in $\mathbf{w}_k$, the cascaded gain factors as
\begin{align}
h_{n,k}&=\mathbf{c}_n^{\top}\mathbf{w}_k,\nonumber\\
[\mathbf{c}_n]_{k'}&=\iint_{\mathcal{A}\times\mathcal{A}}f_n(\mathbf{r})\,\theta(\mathbf{r},\mathbf{r}')\,g_{k'}(\mathbf{r}')\,d\mathbf{r}\,d\mathbf{r}' ,
\label{eq:cascaded_vector}
\end{align}
where $\mathbf{c}_n\in\mathbb{C}^{N}$ is the beamformer-independent cascaded channel of user $n$. This makes explicit that the scalar gain $h_{n,k}$ jointly embeds the BD-RIS response, through $\theta$ in $\mathbf{c}_n$, and the transmit beamforming, through $\mathbf{w}_k$.

\subsection{Signal Model and SINR}
Using \eqref{eq:cascaded_vector}, the signal received by user $n$ separates into the desired stream, the multiuser interference, and noise,
\begin{equation}
y_n=\underbrace{\mathbf{c}_n^{\top}\mathbf{w}_n\,s_n}_{\text{desired}}
+\underbrace{\sum_{k\neq n}\mathbf{c}_n^{\top}\mathbf{w}_k\,s_k}_{\text{interference}}
+z_n,
\label{eq:received_signal}
\end{equation}
where $z_n\sim\mathcal{CN}(0,\sigma_n^2)$ is \ac{AWGN}. The resulting \ac{SINR} at user $n$ is
\begin{equation}
\mathrm{SINR}_n=\frac{|\mathbf{c}_n^{\top}\mathbf{w}_n|^2}{\sum_{k\neq n}|\mathbf{c}_n^{\top}\mathbf{w}_k|^2+\sigma_n^2},
\label{eq:sinr}
\end{equation}
which depends on both the beamformers, through $\{\mathbf{w}_k\}$, and the surface configuration, through the kernel $\theta$ hidden in $\mathbf{c}_n$. The network utility is the downlink sum rate
\begin{equation}
\mathrm{SR}=\sum_{n=1}^{K}\log_2\bigl(1+\mathrm{SINR}_n\bigr),
\label{eq:sumrate_def}
\end{equation}
to be maximized jointly over $\{\mathbf{w}_k\}$ and the surface in Sec.~\ref{sec:sumrate}.

\section{Modal Representation of the Cascaded Channel}
\label{sec:modal}
The cascaded gain \eqref{eq:full_cascaded_continuous} lives in an infinite-dimensional space, yet under the geometric \ac{LoS} model of Sec.~\ref{sec:cont_channel} the incident and radiated fields occupy a low-dimensional subspace of functions on the aperture, generated by finitely many point-source responses. This section makes that structure explicit for an arbitrary orthonormal modal family: it expands the fields on $L\le M$ modes, reduces the scattering operator to a compact $L\times L$ matrix $\boldsymbol{\Psi}$, and connects $\boldsymbol{\Psi}$ to the physical element-domain matrix $\boldsymbol{\Theta}$ that the hardware implements. The specific construction of the modes is deferred to Sec.~\ref{sec:basis}, where the geometry of the deployment will single out one basis and one dimension. Everything in this section and in the optimization of Sec.~\ref{sec:sumrate} holds verbatim for any orthonormal family.

\subsection{Modal Expansion on an Orthonormal Family}
\label{sec:modal_expansion}
We take the aperture to be planar, lying on the plane $\{y=y_{\mathrm R}\}$ and parameterized by its two in-plane coordinates,
\begin{equation}
\mathcal{A}=\bigl\{(x,y,z)\in\mathbb{R}^3:\ y=y_{\mathrm R},\ (x,z)\in\mathcal{A}_{\parallel}\bigr\},
\label{eq:planar_aperture}
\end{equation}
where $\mathcal{A}_{\parallel}\subset\mathbb{R}^2$ is the in-plane domain. Every $\mathbf{r}\in\mathcal{A}$ is thus identified by $\mathbf{r}_{\parallel}=(x,z)\in\mathcal{A}_{\parallel}$, and, the surface being planar, all aperture integrals in \eqref{eq:continuous_scattering} to \eqref{eq:full_cascaded_continuous} reduce to integrals over $\mathcal{A}_{\parallel}$.

Let $\{\varphi_l\}_{l=0}^{L-1}$, $\varphi_l:\mathcal{A}_{\parallel}\to\mathbb{C}$, be any family that is orthonormal with respect to the Hermitian inner product on the aperture,
\begin{equation}
\langle \varphi_l,\varphi_{l'}\rangle
=\int_{\mathcal{A}_{\parallel}}\varphi_l(\mathbf{r}_{\parallel})\,\varphi_{l'}^{*}(\mathbf{r}_{\parallel})\,d\mathbf{r}_{\parallel}=\delta_{l,l'} .
\label{eq:orthonormality}
\end{equation}
We allow the modes to be complex, since several of the physically motivated families of Sec.~\ref{sec:basis}, e.g., spherical wavefronts, are intrinsically complex. Projecting the two fields onto the family,
\begin{align}
e_k(\mathbf{r}_{\parallel})\approx\sum_{l=0}^{L-1}\alpha_{k,l}\,\varphi_l(\mathbf{r}_{\parallel}),\
f_n(\mathbf{r}_{\parallel})\approx\sum_{l=0}^{L-1}\beta_{n,l}\,\varphi_{l}^{*}(\mathbf{r}_{\parallel}),
\label{eq:truncated_fields}
\end{align}
with coefficients $\alpha_{k,l}=\langle e_k,\varphi_l\rangle$ and $\beta_{n,l}=\langle f_n,\varphi_l^{*}\rangle$ stacked into $\boldsymbol{\alpha}_k,\boldsymbol{\beta}_n\in\mathbb{C}^{L}$. The two expansions are deliberately asymmetric: the incident field $e_k$ is expanded on the modes themselves, whereas the radiated functional $f_n$ is expanded on their conjugates. The reason is that in \eqref{eq:full_cascaded_continuous} the user channel $f_n$ enters through the bilinear (not sesquilinear) pairing $\int f_n(\cdot)\,d\mathbf{r}$: as a linear functional, $f_n$ acts as the inner product with $f_n^{*}$, so it is $f_n^{*}$ that must be well captured by $\mathrm{span}\{\varphi_l\}$. This distinction is invisible for real modes but essential for the complex bases of Sec.~\ref{sec:basis}.

Substituting \eqref{eq:truncated_fields} into \eqref{eq:full_cascaded_continuous} yields the bilinear form
\begin{equation}
h_{n,k}\approx\sum_{l,l'}\beta_{n,l}\,\psi_{l,l'}\,\alpha_{k,l'}=\boldsymbol{\beta}_n^{\top}\boldsymbol{\Psi}\,\boldsymbol{\alpha}_k,
\label{eq:bilinear_continuous}
\end{equation}
where the modal scattering matrix $\boldsymbol{\Psi}\in\mathbb{C}^{L\times L}$ collects the projections of the kernel,
\begin{equation}
[\boldsymbol{\Psi}]_{l,l'}=\psi_{l,l'}=\iint_{\mathcal{A}_{\parallel}\times\mathcal{A}_{\parallel}}\varphi_{l}^{*}(\mathbf{r})\,\theta(\mathbf{r},\mathbf{r}')\,\varphi_{l'}(\mathbf{r}')\,d\mathbf{r}\,d\mathbf{r}' .
\label{eq:physical_kernel}
\end{equation}
Conversely, any kernel of the modal form
\begin{equation}
\theta(\mathbf{r},\mathbf{r}')=\sum_{l,l'}\varphi_{l}(\mathbf{r})\,[\boldsymbol{\Psi}]_{l,l'}\,\varphi_{l'}^{*}(\mathbf{r}')
\label{eq:modal_kernel}
\end{equation}
reproduces \eqref{eq:physical_kernel} exactly: substituting \eqref{eq:modal_kernel} back and invoking the orthonormality \eqref{eq:orthonormality} collapses the double sum through the Kronecker deltas,
\begin{equation}
\psi_{l,l'}=\sum_{m,m'}\psi_{m,m'}\,\delta_{l,m}\,\delta_{l',m'} .
\label{eq:Q_phys_modal}
\end{equation}
The entries of $\boldsymbol{\Psi}$ are thus the beam-domain coupling coefficients, exactly as the entries of a classical BD-RIS scattering matrix are the element-domain ones. The consistency \eqref{eq:Q_phys_modal} concerns only the modal coefficients, however, and does not by itself carry over the passivity constraint \eqref{eq:continuous_unitarity}: as shown in Sec.~\ref{sec:trunc}, the truncated kernel \eqref{eq:modal_kernel} is a partial isometry and must be completed on the orthogonal complement to remain lossless.

\subsection{From the Continuous Aperture to Discrete Elements}
\label{sec:discretization}
The physical surface samples the aperture at $M$ elements located at $\{\mathbf{r}_m\}_{m=1}^{M}$ on a $\lambda/2$ grid. Sampling \eqref{eq:incident_field} gives the discrete incident field
\begin{equation}
[\mathbf{e}_k]_m=e_k(\mathbf{r}_m)=\sum_{k'=1}^{N}w_{k,k'}\,g_{k'}(\mathbf{r}_m),\qquad m=1,\dots,M,
\label{eq:discrete_incident}
\end{equation}
or, compactly, $\mathbf{e}_k=\mathbf{G}\mathbf{w}_k\in\mathbb{C}^{M}$, with $[\mathbf{G}]_{m,k'}=g_{k'}(\mathbf{r}_m)$, $\mathbf{G}\in\mathbb{C}^{M\times N}$, the sampled \ac{BS}-RIS matrix.
At the element level the surface is described by its element-domain scattering matrix $\boldsymbol{\Theta}\in\mathbb{C}^{M\times M}$, the discrete counterpart of the kernel $\theta$, which is what the reconfigurable impedance network physically realizes. The RIS-user response in \eqref{eq:ris_user_scattered_channel} and the reflected field in \eqref{eq:continuous_scattering} are sampled as $\mathbf{f}_n\in\mathbb{C}^{M}$ and $\mathbf{u}_k=\boldsymbol{\Theta}\mathbf{e}_k$, respectively. Therefore, the continuous cascaded gain defined in \eqref{eq:full_cascaded_continuous} becomes
\begin{equation}
h_{n,k}=\mathbf{f}_n^{\top}\boldsymbol{\Theta}\,\mathbf{G}\mathbf{w}_k=\mathbf{f}_n^{\top}\boldsymbol{\Theta}\,\mathbf{e}_k .
\label{eq:discrete_cascaded}
\end{equation}
Likewise, sampling the modal family at the element positions and stacking the modes columnwise yields the modal matrix
\begin{equation}
\boldsymbol{\Phi}=[\,\boldsymbol{\varphi}_0,\boldsymbol{\varphi}_1,\dots,\boldsymbol{\varphi}_{L-1}\,]\in\mathbb{C}^{M\times L},\quad
\boldsymbol{\Phi}^{\mathsf H}\boldsymbol{\Phi}=\mathbf{I}_L,
\label{eq:Phi_def}
\end{equation}
where $[\boldsymbol{\varphi}_l]_m=\varphi_l(\mathbf{r}_{\parallel,m})$ up to a common grid-normalization factor. The discrete orthonormality in \eqref{eq:Phi_def} either is inherited exactly from \eqref{eq:orthonormality} (as for the \ac{DFT} modes on a uniform grid) or is enforced by a thin QR factorization of the sampled family, which leaves the spanned subspace unchanged. Since the dimension of the sampled space is $M$, any family satisfying \eqref{eq:Phi_def} has $1\le L\le M$, and at $L=M$ it is a complete orthonormal basis of $\mathbb{C}^{M}$: the modal description then becomes a mere change of coordinates of the element domain. The two fields are represented by their modal coordinates
\begin{align}
\boldsymbol{\alpha}_k&=\boldsymbol{\Phi}^{\mathsf H}\mathbf{e}_k=\boldsymbol{\Phi}^{\mathsf H}\mathbf{G}\mathbf{w}_k\in\mathbb{C}^{L}, \label{eq:e_modal_approx}\\
\boldsymbol{\beta}_n&=\boldsymbol{\Phi}^{\top}\mathbf{f}_n\in\mathbb{C}^{L}, \label{eq:beta_from_f}
\end{align}
which are the discrete counterparts of the coefficients in \eqref{eq:truncated_fields}. Note again the transpose (not Hermitian) projection of $\mathbf{f}_n$, consistent with the bilinear pairing in \eqref{eq:discrete_cascaded}. The surface is designed through the compact modal matrix $\boldsymbol{\Psi}$, linked to the element-domain scattering matrix by the compression
\begin{equation}
\boldsymbol{\Psi}=\boldsymbol{\Phi}^{\mathsf H}\boldsymbol{\Theta}\boldsymbol{\Phi}\in\mathbb{C}^{L\times L}.
\label{eq:psi_theta_link}
\end{equation}
The two bilinear forms \eqref{eq:discrete_cascaded} and \eqref{eq:bilinear_continuous} describe the same end-to-end gain in two representations: in the modal domain via $\boldsymbol{\Psi}$ and in the element domain through $\boldsymbol{\Theta}$. Only the low-dimensional $\boldsymbol{\Psi}$ is optimized.

\subsection{Lossless Completion and Effective Channel}
\label{sec:trunc}
The naive inverse of \eqref{eq:psi_theta_link}, $\boldsymbol{\Theta}=\boldsymbol{\Phi}\boldsymbol{\Psi}\boldsymbol{\Phi}^{\mathsf H}$, is not admissible: it has rank $L$ and, for $L<M$, violates the discrete image $\boldsymbol{\Theta}\boldsymbol{\Theta}^{\mathsf H}=\mathbf{I}_M$ of the passivity condition \eqref{eq:continuous_unitarity}. Indeed, using $\boldsymbol{\Phi}^{\mathsf H}\boldsymbol{\Phi}=\mathbf{I}_L$ and $\boldsymbol{\Psi}\boldsymbol{\Psi}^{\mathsf H}=\mathbf{I}_L$,
\begin{equation}
\boldsymbol{\Phi}\boldsymbol{\Psi}\boldsymbol{\Phi}^{\mathsf H}\bigl(\boldsymbol{\Phi}\boldsymbol{\Psi}\boldsymbol{\Phi}^{\mathsf H}\bigr)^{\mathsf H}
=\boldsymbol{\Phi}\boldsymbol{\Psi}\underbrace{\boldsymbol{\Phi}^{\mathsf H}\boldsymbol{\Phi}}_{\mathbf{I}_L}\boldsymbol{\Psi}^{\mathsf H}\boldsymbol{\Phi}^{\mathsf H}
=\boldsymbol{\Phi}\boldsymbol{\Phi}^{\mathsf H}=\boldsymbol{\Pi}\neq\mathbf{I}_M,
\label{eq:partial_isometry}
\end{equation}
with $\boldsymbol{\Pi}=\boldsymbol{\Phi}\boldsymbol{\Phi}^{\mathsf H}$ the orthogonal projector onto $\mathrm{range}(\boldsymbol{\Phi})$: the truncated surface is a partial isometry, lossless on the modal subspace and absorbing on its complement. Passivity is restored by completing $\boldsymbol{\Theta}$ with an identity action on the complementary modes,
\begin{equation}
\boldsymbol{\Theta}=\boldsymbol{\Phi}\boldsymbol{\Psi}\boldsymbol{\Phi}^{\mathsf H}+\bigl(\mathbf{I}_M-\boldsymbol{\Phi}\boldsymbol{\Phi}^{\mathsf H}\bigr),
\label{eq:theta_lossless}
\end{equation}
which is unitary for every $\boldsymbol{\Psi}\in\mathcal{U}(L)$: since $\boldsymbol{\Phi}^{\mathsf H}(\mathbf{I}_M-\boldsymbol{\Phi}\boldsymbol{\Phi}^{\mathsf H})=\mathbf{0}$ annihilates the cross terms, $\boldsymbol{\Theta}\boldsymbol{\Theta}^{\mathsf H}=\boldsymbol{\Pi}+(\mathbf{I}_M-\boldsymbol{\Pi})=\mathbf{I}_M$, and \eqref{eq:psi_theta_link} is recovered by left/right multiplication with $\boldsymbol{\Phi}^{\mathsf H}$/$\boldsymbol{\Phi}$.

The matrix \eqref{eq:theta_lossless} is dense, so $L^{2}$ counts the reconfigurable degrees of freedom, not the physical ports: the surface realizes $\boldsymbol{\Theta}$ as a fixed passive mode-forming network implementing $\boldsymbol{\Phi}$, an $L\times L$ tunable core implementing $\boldsymbol{\Psi}$, and a fixed completion, so that only $L^{2}$ entries are adjusted at run time.

Substituting \eqref{eq:theta_lossless} into \eqref{eq:discrete_cascaded} and using \eqref{eq:e_modal_approx} and \eqref{eq:beta_from_f}, the cascaded gain splits into a modal term and a complement term,
\begin{align}
h_{n,k}=\mathbf{f}_n^{\top}\boldsymbol{\Theta}\mathbf{G}\mathbf{w}_k
&=\boldsymbol{\beta}_n^{\top}\boldsymbol{\Psi}\boldsymbol{\Phi}^{\mathsf H}\mathbf{G}\mathbf{w}_k
+\mathbf{f}_n^{\top}\bigl(\mathbf{I}_M-\boldsymbol{\Phi}\boldsymbol{\Phi}^{\mathsf H}\bigr)\mathbf{G}\mathbf{w}_k\nonumber\\
&=\boldsymbol{\gamma}_n^{\top}\mathbf{w}_k,
\label{eq:g_eff_explicit}
\end{align}
so that the effective beam-domain channel carries a fixed offset,
\begin{align}
\boldsymbol{\gamma}_n&=\mathbf{G}^{\top}\boldsymbol{\Phi}^{*}\boldsymbol{\Psi}^{\top}\boldsymbol{\beta}_n+\boldsymbol{\gamma}_n^{0},\nonumber\\
\boldsymbol{\gamma}_n^{0}&=\mathbf{G}^{\top}\bigl(\mathbf{I}_M-\boldsymbol{\Phi}^{*}\boldsymbol{\Phi}^{\top}\bigr)\mathbf{f}_n\in\mathbb{C}^{N},
\label{eq:g_eff_def}
\end{align}
where the conjugates arise from transposing the Hermitian projector, $(\mathbf{I}_M-\boldsymbol{\Phi}\boldsymbol{\Phi}^{\mathsf H})^{\top}=\mathbf{I}_M-\boldsymbol{\Phi}^{*}\boldsymbol{\Phi}^{\top}$, and disappear for real $\boldsymbol{\Phi}$. The first term is the discrete, modal counterpart of the continuous $\mathbf{c}_n$ in \eqref{eq:cascaded_vector}. The offset $\boldsymbol{\gamma}_n^{0}$ is the mirror reflection of the field components outside $\mathrm{range}(\boldsymbol{\Phi})$, is independent of $\boldsymbol{\Psi}$, and is precomputed once. The design thus collapses the surface to the $L\ll M$ degrees of freedom of $\boldsymbol{\Psi}$ while keeping \eqref{eq:g_eff_explicit} exact and the surface exactly passive. Setting $\boldsymbol{\gamma}_n^{0}=\mathbf{0}$ recovers the lossy rank-$L$ model, exact only when the fields lie entirely in $\mathrm{range}(\boldsymbol{\Phi})$.

\begin{remark}
The truncated kernel \eqref{eq:modal_kernel} alone yields the partial isometry $\boldsymbol{\Theta}=\boldsymbol{\Phi}\boldsymbol{\Psi}\boldsymbol{\Phi}^{\mathsf H}$, for which $\boldsymbol{\Theta}\boldsymbol{\Theta}^{\mathsf H}=\boldsymbol{\Pi}\neq\mathbf{I}_M$, so passivity \eqref{eq:continuous_unitarity} holds only at $L=M$. The completion \eqref{eq:theta_lossless} makes $\boldsymbol{\Theta}$ unitary for every $\boldsymbol{\Psi}\in\mathcal{U}(L)$, independently of how the incident energy is distributed: the design $\boldsymbol{\Psi}$ shapes the response inside $\mathrm{range}(\boldsymbol{\Phi})$ while the complementary modes are reflected unchanged. Equivalently, restricting the unitary-manifold search direction to the modal subspace yields $\exp(-\eta\,\boldsymbol{\Phi}\boldsymbol{\Xi}\boldsymbol{\Phi}^{\mathsf H})=\boldsymbol{\Phi}e^{-\eta\boldsymbol{\Xi}}\boldsymbol{\Phi}^{\mathsf H}+(\mathbf{I}_M-\boldsymbol{\Phi}\boldsymbol{\Phi}^{\mathsf H})$, so a geodesic step of $\boldsymbol{\Theta}$ on $\mathcal{U}(M)$ confined to $\mathrm{range}(\boldsymbol{\Phi})$ reduces exactly to the $L\times L$ update of Sec.~\ref{sec:modal_update}.
\end{remark}

\section{The Geometric Modal Basis}
\label{sec:basis}
Sections~\ref{sec:modal_expansion} to \ref{sec:trunc} reduced the surface design to the choice of an orthonormal family $\{\varphi_l\}$, equivalently of the subspace $\mathrm{range}(\boldsymbol{\Phi})$ onto which the incident and radiated fields are projected. Under the \ac{LoS} model of Sec.~\ref{sec:cont_channel} this choice admits a complete and remarkably compact answer, developed in this section: the fields on the aperture live in the span of $N{+}K$ known spherical-wave atoms, the atoms are determined by the terminal positions alone, and a basis of dimension $L=2\rho$ built from them provably renders the modal architecture equivalent to an unconstrained $M\times M$ surface.

\subsection{Spherical-Wave Atoms and the Active Subspace}
\label{sec:atoms}
Collecting the $M$ samples of the aperture response to a point source at $\mathbf{s}\in\mathbb{R}^3$ into the spherical-wave atom
\begin{equation}
[\mathbf{a}(\mathbf{s})]_m=\mathcal{G}(\mathbf{r}_m,\mathbf{s})=\frac{e^{-j\kappa_0\|\mathbf{r}_m-\mathbf{s}\|}}{\|\mathbf{r}_m-\mathbf{s}\|},
\label{eq:atom}
\end{equation}
the sampled channel matrices are, by \eqref{eq:bs_ris_scattered_channel} and \eqref{eq:ris_user_scattered_channel}, exactly collections of atoms at the terminal positions,
\begin{equation}
\mathbf{G}=\sqrt{G_{\mathrm{BS}}}\,\bigl[\mathbf{a}(\mathbf{b}_1),\dots,\mathbf{a}(\mathbf{b}_N)\bigr],\quad
\mathbf{F}=\bigl[\mathbf{a}(\mathbf{u}_1),\dots,\mathbf{a}(\mathbf{u}_K)\bigr].
\label{eq:GF_atoms}
\end{equation}
Since the cascaded gains \eqref{eq:discrete_cascaded} can be written $h_{n,k}=\mathbf{f}_n^{\top}\boldsymbol{\Theta}\mathbf{e}_k=(\mathbf{f}_n^{*})^{\mathsf H}\boldsymbol{\Theta}\mathbf{e}_k$, with $\mathbf{e}_k\in\mathrm{range}(\mathbf{G})$ and $\mathbf{f}_n^{*}\in\mathrm{range}(\mathbf{F}^{*})$, they depend on $\boldsymbol{\Theta}$ only through its compression to the active subspace
\begin{equation}
\mathcal{S}=\mathrm{span}\bigl([\mathbf{G},\ \mathbf{F}^{*}]\bigr),\qquad \rho=\dim\mathcal{S}\le N+K,
\label{eq:active_subspace}
\end{equation}
namely the contraction $\mathbf{T}=\boldsymbol{\Pi}_{\mathcal{S}}\boldsymbol{\Theta}\boldsymbol{\Pi}_{\mathcal{S}}$, with $\boldsymbol{\Pi}_{\mathcal{S}}$ the projector onto $\mathcal{S}$ and $\|\mathbf{T}\|_2\le1$ because $\boldsymbol{\Theta}$ is unitary. This contraction is the pivotal object of the paper: since the rates depend on $\boldsymbol{\Theta}$ only through $\mathbf{T}$, it suffices to realize every admissible $\mathbf{T}$, and Proposition~\ref{prop:halmos} shows that a unitary of dimension $2\rho$ is enough to do so exactly. The conjugation of $\mathbf{F}$ is the discrete image of the bilinear (not sesquilinear) pairing already noted in \eqref{eq:truncated_fields}. Two structural facts sharpen \eqref{eq:active_subspace}. First, the atoms in \eqref{eq:GF_atoms} are position-determined: constructing a basis for $\mathcal{S}$ requires knowledge of where the terminals are, not of any instantaneous channel state. Under the quasi-static positions assumed throughout, the basis is a deployment quantity, recomputed on the localization timescale. Second, the effective dimension of $\mathrm{range}(\mathbf{G})$ is limited not only by the antenna count but by the near-field coupling geometry: the numerical rank of the $N$ \ac{BS}-side atoms saturates at the aperture-product degrees of freedom, $r_G=\min\bigl(N,\;\eta_{\mathrm{dof}}\,D_{\mathrm{BS}}D/(\lambda\, d_{\mathrm{BS}})\bigr)$ with $\eta_{\mathrm{dof}}$ a geometry constant of order one, so that in practice $\rho=r_G+K$ with $r_G$ read off the deployment. The $K$ user atoms, sampled at well-separated positions across the near field of an extremely large panel, are nearly orthogonal and contribute their full count.

\subsection{Basis Construction}
\label{sec:construction}
Two facts delimit what any basis can achieve. First, only the subspace matters, not the particular family spanning it.

\begin{proposition}[Subspace dependence]
\label{prop:subspace}
Let $\boldsymbol{\Phi}'=\boldsymbol{\Phi}\mathbf{O}$ with $\mathbf{O}\in\mathcal{U}(L)$. Then the sets of element-domain surfaces $\{\boldsymbol{\Theta}\}$ generated by \eqref{eq:theta_lossless} under $\boldsymbol{\Phi}$ and $\boldsymbol{\Phi}'$ coincide, and hence so do all achievable rate tuples. The performance of the modal architecture at dimension $L$ depends on the basis only through the projector $\boldsymbol{\Pi}=\boldsymbol{\Phi}\boldsymbol{\Phi}^{\mathsf H}$.
\end{proposition}
\begin{proof}
$\boldsymbol{\Phi}'\boldsymbol{\Psi}'(\boldsymbol{\Phi}')^{\mathsf H}=\boldsymbol{\Phi}(\mathbf{O}\boldsymbol{\Psi}'\mathbf{O}^{\mathsf H})\boldsymbol{\Phi}^{\mathsf H}$ and $\boldsymbol{\Psi}'\mapsto\mathbf{O}\boldsymbol{\Psi}'\mathbf{O}^{\mathsf H}$ is a bijection of $\mathcal{U}(L)$. The complement term $\mathbf{I}_M-\boldsymbol{\Pi}$ in \eqref{eq:theta_lossless} depends only on $\boldsymbol{\Pi}$.
\end{proof}

By Proposition~\ref{prop:subspace} the design target is a subspace, and by \eqref{eq:active_subspace} the subspace that matters is $\mathcal{S}$. The geometric modal basis is therefore obtained directly from the stacked atom matrix: compute the thin \ac{SVD}
\begin{equation}
[\mathbf{G},\ \mathbf{F}^{*}]=\mathbf{U}\boldsymbol{\Sigma}\mathbf{V}^{\mathsf H},\qquad
\boldsymbol{\Phi}_{\mathcal{S}}=\mathbf{U}(:,1{:}\rho),
\label{eq:svd_construction}
\end{equation}
with $\rho$ the numerical rank of $\boldsymbol{\Sigma}$ at a fixed relative tolerance\footnote{In practice $\rho$ can be $\leq N{+}K$ when the atoms are correlated, which only shrinks the basis further: the construction is exact on the $\epsilon$-truncated channel, the discarded singular values negligibly perturb the rates.}, and complete it with an arbitrary orthonormal complement $\boldsymbol{\Phi}_{\perp}\in\mathbb{C}^{M\times\rho}$, $\boldsymbol{\Phi}_{\perp}^{\mathsf H}\boldsymbol{\Phi}_{\mathcal{S}}=\mathbf{0}$,
\begin{equation}
\boldsymbol{\Phi}=[\,\boldsymbol{\Phi}_{\mathcal{S}},\ \boldsymbol{\Phi}_{\perp}\,]\in\mathbb{C}^{M\times L},\qquad L=2\rho .
\label{eq:phi_completion}
\end{equation}
The role of the second block is purely to provide unitary slack: by Proposition~\ref{prop:subspace} its particular choice is immaterial, and the doubling $L=2\rho$ is exactly what Proposition~\ref{prop:halmos} below consumes. The construction costs one thin \ac{SVD}, $O\bigl(M(N{+}K)^2\bigr)$, performed once per localization update.

\begin{remark}[Covariance formulation and computation]
\label{rem:covariance}
Writing $\mathbf{A}=[\mathbf{G},\,\mathbf{F}^{*}]\in\mathbb{C}^{M\times(N+K)}$ for the atom stack, the basis \eqref{eq:svd_construction} admits an equivalent second-order description: $\boldsymbol{\Phi}_{\mathcal{S}}$ is the dominant eigenspace of the aperture correlation matrix
\begin{align}
\mathbf{R}=\mathbf{A}\mathbf{A}^{\mathsf H}
&=\sum_{\mathbf{s}\in\{\mathbf{b}_k\}} G_{\mathrm{BS}}\,\mathbf{a}(\mathbf{s})\mathbf{a}(\mathbf{s})^{\mathsf H}
+\sum_{\mathbf{s}\in\{\mathbf{u}_n\}} \mathbf{a}^{*}(\mathbf{s})\bigl(\mathbf{a}^{*}(\mathbf{s})\bigr)^{\mathsf H},
\label{eq:correlation_operator}
\end{align}
i.e., the correlation operator of the aperture field under the empirical measure placed on the realized terminal positions, with conjugated atoms on the user side reflecting the bilinear pairing. The eigenvectors of $\mathbf{R}$ are the left singular vectors of $\mathbf{A}$. Since $\mathbf{A}\mathbf{A}^{\mathsf H}$ and the $(N{+}K)\times(N{+}K)$ Gram matrix $\mathbf{A}^{\mathsf H}\mathbf{A}$ share their nonzero spectrum $\{\sigma_i^2\}$, the basis is recovered from the small Hermitian eigenproblem $\mathbf{A}^{\mathsf H}\mathbf{A}=\mathbf{V}\boldsymbol{\Sigma}^{2}\mathbf{V}^{\mathsf H}$ as $\boldsymbol{\Phi}_{\mathcal{S}}=\mathbf{A}\,\mathbf{V}(:,1{:}\rho)\,\boldsymbol{\Sigma}_{1:\rho}^{-1}$ which is precisely what the thin \ac{SVD} in \eqref{eq:svd_construction} computes. Either route reduces the construction to orthonormalizing $N{+}K$ analytically known spherical-wave vectors, rank-revealed by the spectrum of an $(N{+}K)\times(N{+}K)$ Hermitian matrix, at total cost $O\bigl(M(N{+}K)^2\bigr)$. The covariance view indicates the natural generalization: replacing the empirical measure in \eqref{eq:correlation_operator} by a distribution over position-uncertainty regions (or over scattered multipath, when present) turns $\mathbf{R}$ into a genuine ensemble covariance whose dominant eigenspace is the Karhunen-Lo\`eve basis, mean-square optimal at every $L$ by the Ky Fan principle. The geometric basis used here is the deterministic endpoint of that family, and its robustness to position uncertainty is quantified through it.
\end{remark}

\subsection{Exactness at $L=2\rho$}
\label{sec:criteria}
The dimension $L=2\rho$ in \eqref{eq:phi_completion} is exactly what attaining fully-connected performance requires.

\begin{proposition}[Exactness of the geometric basis]
\label{prop:halmos}
If $\mathcal{S}\subseteq\mathrm{range}(\boldsymbol{\Phi})$ and $L\ge 2\rho$, then
\begin{equation}
\max_{\boldsymbol{\Psi}\in\,\mathcal{U}(L)}\mathrm{SR}\bigl(\boldsymbol{\Theta}(\boldsymbol{\Psi})\bigr)
=\max_{\boldsymbol{\Theta}\in\,\mathcal{U}(M)}\mathrm{SR}(\boldsymbol{\Theta}),
\end{equation}
for the sum rate \eqref{eq:sumrate_def} maximized over the beamformers under the power budget and, more generally, for any utility depending on the surface only through the gains $\{h_{n,k}\}$. In particular, the basis \eqref{eq:phi_completion} attains the fully-connected optimum at $L=2\rho\le 2(N{+}K)$, independently of $M$.
\end{proposition}
\begin{proof}
The inequality $(\le)$ holds because $\boldsymbol{\Theta}(\boldsymbol{\Psi})\in\mathcal{U}(M)$ by \eqref{eq:theta_lossless}. For $(\ge)$, fix any $\boldsymbol{\Theta}_{\star}\in\mathcal{U}(M)$ and let $\mathbf{T}=\boldsymbol{\Pi}_{\mathcal{S}}\boldsymbol{\Theta}_{\star}\boldsymbol{\Pi}_{\mathcal{S}}$ be its compression to $\mathcal{S}$, a contraction on a $\rho$-dimensional space. By the Halmos unitary dilation \cite{halmos1950normal}, the $2\rho\times 2\rho$ matrix
\begin{equation}
\mathbf{U}_{\mathbf{T}}=\begin{bmatrix}\mathbf{T} & (\mathbf{I}-\mathbf{T}\mathbf{T}^{\mathsf H})^{1/2}\\ (\mathbf{I}-\mathbf{T}^{\mathsf H}\mathbf{T})^{1/2} & -\mathbf{T}^{\mathsf H}\end{bmatrix}
\label{eq:halmos}
\end{equation}
is unitary with upper-left $\rho\times\rho$ block $\mathbf{T}$. Since $L\ge 2\rho$ and $\mathcal{S}\subseteq\mathrm{range}(\boldsymbol{\Phi})$, pick an isometry $\mathbf{J}\in\mathbb{C}^{L\times 2\rho}$ whose first $\rho$ columns are the coordinates of an orthonormal basis of $\mathcal{S}$ in $\mathrm{range}(\boldsymbol{\Phi})$ and whose remaining $\rho$ columns extend it orthonormally within those coordinates (possible since $L\ge2\rho$), and set $\boldsymbol{\Psi}=\mathbf{J}\mathbf{U}_{\mathbf{T}}\mathbf{J}^{\mathsf H}+(\mathbf{I}_L-\mathbf{J}\mathbf{J}^{\mathsf H})\in\mathcal{U}(L)$. Then $\boldsymbol{\Theta}(\boldsymbol{\Psi})$ has the same compression $\mathbf{T}$ on $\mathcal{S}$. Since $\mathbf{f}_n^{*}\in\mathcal{S}$ and $\mathbf{e}_k=\mathbf{G}\mathbf{w}_k\in\mathcal{S}$ for every beamformer, the effective channels $\mathbf{f}_n^{\top}\boldsymbol{\Theta}(\boldsymbol{\Psi})\mathbf{G}$ and $\mathbf{f}_n^{\top}\boldsymbol{\Theta}_{\star}\mathbf{G}$ coincide as row vectors for all $n$. Hence all gains $h_{n,k}$, the utility at any fixed $\{\mathbf{w}_k\}$, and in particular its maximum over the beamformers, match those of $\boldsymbol{\Theta}_{\star}$.
\end{proof}

Proposition~\ref{prop:halmos} is the load-bearing statement of the paper: it converts the $M$-dimensional surface design into an exactly equivalent problem of dimension $2\rho\le 2(N{+}K)$, with $M$ entering only through the one-time compression of Sec.~\ref{sec:construction}. The next section solves that problem. Sec.~\ref{sec:architectures} then situates the result against the conventional element-domain architectures, whose cost-performance frontier it strictly dominates.
\section{Sum-Rate Maximization}
\label{sec:sumrate}
We jointly design the beamformers $\{\mathbf{w}_k\}_{k=1}^{K}$ and the modal configuration $\boldsymbol{\Psi}$ to maximize the downlink sum spectral efficiency:
\begin{subequations}
\begin{align}
\max_{\{\mathbf{w}_k\},\,\boldsymbol{\Psi}}\ &\ \mathrm{SR}\triangleq\sum_{n=1}^{K}\log_2\!\bigl(1+\mathrm{SINR}_n\bigr) \label{eq:sumrate_obj_modal}\\
\mathrm{s.t.}\ &\ \boldsymbol{\Psi}^{\mathsf{H}}\boldsymbol{\Psi}=\mathbf{I}_L, \label{eq:psi_unitary_modal}\\
&\ \sum_{k=1}^{K}\|\mathbf{w}_k\|^2\le P, \label{eq:power_constraint_modal}
\end{align}
\end{subequations}
where the \ac{SINR} is now $\mathrm{SINR}_n=|\boldsymbol{\gamma}_n^{\top}\mathbf{w}_n|^2/(\sum_{k\neq n}|\boldsymbol{\gamma}_n^{\top}\mathbf{w}_k|^2+\sigma_n^2)$ with $\boldsymbol{\gamma}_n$ from \eqref{eq:g_eff_def}. This is exactly the \ac{SINR} \eqref{eq:sinr} of the original signal model, rewritten through the exact split \eqref{eq:g_eff_explicit}, so no approximation is introduced by the modal formulation. Constraint \eqref{eq:psi_unitary_modal} keeps the surface passive, by \eqref{eq:theta_lossless}, and \eqref{eq:power_constraint_modal} is the transmit power budget. The formulation holds for any orthonormal $\boldsymbol{\Phi}$, in particular the geometric basis of Sec.~\ref{sec:basis}, entering only through the precomputed quantities $\mathbf{G}^{\top}\boldsymbol{\Phi}^{*}$, $\{\boldsymbol{\beta}_n\}$, and $\{\boldsymbol{\gamma}_n^{0}\}$. The problem is non-convex, owing to the bilinear coupling of $\{\mathbf{w}_k\}$ and $\boldsymbol{\Psi}$ inside the \ac{SINR} and to the unitary manifold constraint. We address it by \ac{AO}, alternating a \ac{WMMSE} update of the beamformers with a Riemannian update of $\boldsymbol{\Psi}$.

\subsection{WMMSE Reformulation}
For each user we introduce a receive scalar $v_n\in\mathbb{C}$ and a weight $q_n>0$. The \ac{MSE} of the soft estimate $v_n^{*}y_n$ is
\begin{align}
\mathrm{MSE}_n&=\mathbb{E}\bigl[|v_n^{*}y_n-s_n|^2\bigr]\nonumber\\
&=|v_n|^2\Bigl(\sum_{k=1}^{K}|\boldsymbol{\gamma}_n^{\top}\mathbf{w}_k|^2+\sigma_n^2\Bigr)-2\Re\{v_n^{*}\boldsymbol{\gamma}_n^{\top}\mathbf{w}_n\}+1 .
\label{eq:mse_n_reform}
\end{align}
Minimizing $\mathrm{MSE}_n$ over $v_n$ gives the \ac{MMSE} receiver
\begin{equation}
v_n^{\star}=\frac{\boldsymbol{\gamma}_n^{\top}\mathbf{w}_n}{\sum_{k=1}^{K}|\boldsymbol{\gamma}_n^{\top}\mathbf{w}_k|^2+\sigma_n^2},
\label{eq:mmse_receiver}
\end{equation}
whose optimal value satisfies $\mathrm{MSE}_n^{\star}=1/(1+\mathrm{SINR}_n)$. Defining the weighted \ac{MSE}
\begin{equation}
\mathrm{WMSE}_n=q_n\,\mathrm{MSE}_n-\ln q_n,
\label{eq:wmse_n_reform}
\end{equation}
its minimizer over $q_n$ is $q_n^{\star}=1/\mathrm{MSE}_n^{\star}=1+\mathrm{SINR}_n$, at which $\mathrm{WMSE}_n^{\star}=1-\ln(1+\mathrm{SINR}_n)$. Since maximizing the sum rate in bits or nats differs only by the constant $\ln 2$, the design \eqref{eq:sumrate_obj_modal} is equivalent to
\begin{equation}
\min_{\{\mathbf{w}_k\},\,\{v_n,q_n\},\,\boldsymbol{\Psi}}\ \mathcal{L}\triangleq\sum_{n=1}^{K}\mathrm{WMSE}_n
\quad\mathrm{s.t.}\quad\eqref{eq:psi_unitary_modal},\ \eqref{eq:power_constraint_modal}.
\label{eq:wmmse_equiv}
\end{equation}
For fixed $\{v_n,q_n\}$ the objective is quadratic in $\{\mathbf{w}_k\}$ and in $\boldsymbol{\Psi}$, which is what enables the closed-form and manifold updates below.

\subsection{Beamforming Update}
With $\boldsymbol{\Psi}$ and $\{v_n,q_n\}$ fixed, and stacking the beamformers into $\mathbf{W}=[\mathbf{w}_1,\dots,\mathbf{w}_K]\in\mathbb{C}^{N\times K}$, the $\boldsymbol{\Psi}$-independent part of \eqref{eq:wmmse_equiv} becomes
\begin{align}
&\mathcal{L}(\mathbf{W})=\sum_{n=1}^{K}q_n\Bigl(|v_n|^2\sum_{k}|\boldsymbol{\gamma}_n^{\top}\mathbf{w}_k|^2-2\Re\{v_n^{*}\boldsymbol{\gamma}_n^{\top}\mathbf{w}_n\}\Bigr)\nonumber\\
&\overset{(a)}{=}\sum_{n=1}^{K}q_n\Bigl(|v_n|^2\sum_{k}\boldsymbol{\epsilon}_k^{\top}\mathbf{W}^{\mathsf{H}}\boldsymbol{\gamma}_n^{*}\boldsymbol{\gamma}_n^{\top}\mathbf{W}\boldsymbol{\epsilon}_k-2\Re\{\boldsymbol{\epsilon}_n^{\mathsf{H}}\mathbf{W}^{\mathsf{H}}\boldsymbol{\gamma}_n^{*}v_n\}\Bigr)\nonumber\\
&\overset{(b)}{=}\mathrm{Tr}\bigl(\mathbf{W}^{\mathsf{H}}\mathbf{A}\mathbf{W}\bigr)-2\Re\bigl\{\mathrm{Tr}\bigl(\mathbf{W}^{\mathsf{H}}\mathbf{B}\bigr)\bigr\},
\label{eq:LW_quadratic}
\end{align}
with
\begin{equation}
\mathbf{A}=\sum_{n}q_n|v_n|^2\,\boldsymbol{\gamma}_n^{*}\boldsymbol{\gamma}_n^{\top}\succeq0,\qquad
\mathbf{B}=\sum_{n}q_n v_n\,\boldsymbol{\gamma}_n^{*}\boldsymbol{\epsilon}_n^{\mathsf{H}},
\label{eq:AB_def}
\end{equation}
where $\boldsymbol{\epsilon}_n$ is the $n$-th canonical vector. Step $(a)$ uses $\mathbf{w}_k=\mathbf{W}\boldsymbol{\epsilon}_k$, and step $(b)$ uses $\sum_k\boldsymbol{\epsilon}_k^{\top}(\cdot)\boldsymbol{\epsilon}_k=\mathrm{Tr}(\cdot)$ together with the circular trace property. The power constraint \eqref{eq:power_constraint_modal} is enforced through a multiplier $\mu\ge0$, and setting $\nabla_{\mathbf{W}^{*}}\mathcal{L}=\mathbf{0}$ yields
\begin{equation}
\mathbf{W}^{\star}=(\mathbf{A}+\mu\mathbf{I}_N)^{-1}\mathbf{B},
\label{eq:W_closedform_deriv}
\end{equation}
where $\mu$ is the smallest nonnegative value for which $\mathrm{Tr}\bigl((\mathbf{W}^{\star})^{\mathsf{H}}\mathbf{W}^{\star}\bigr)\le P$. The value $\mu=0$ applies whenever the unconstrained solution $\mathbf{A}^{-1}\mathbf{B}$ already meets the budget, and \eqref{eq:W_closedform_deriv} is then the exact minimizer of \eqref{eq:LW_quadratic} under \eqref{eq:power_constraint_modal}. In practice we diagonalize $\mathbf{A}=\mathbf{U}\mathrm{diag}(\mathbf{a})\mathbf{U}^{\mathsf H}$ once, after which the power $\mathrm{Tr}((\mathbf{W}^{\star})^{\mathsf H}\mathbf{W}^{\star})=\sum_{i}p_i/(a_i+\mu)^2$ is a scalar, monotonically decreasing function of $\mu$, and the multiplier is found by bisection at $O(K)$ per evaluation, without ever solving a near-singular system.

The exact minimizer may use less than the full budget. The following elementary fact shows that projecting it radially onto the power sphere is not a heuristic but a guaranteed ascent step on the original objective.
\begin{lemma}[Radial power monotonicity]
\label{lem:radial}
Fix the surface, hence the effective channels $\{\boldsymbol{\gamma}_n\}$. For any beamformers $\mathbf{W}$ with $\|\mathbf{W}\|_{\mathrm F}^{2}\le P$ and any $\alpha\ge 1$ with $\alpha^2\|\mathbf{W}\|_{\mathrm F}^{2}\le P$, the sum rate \eqref{eq:sumrate_def} satisfies $\mathrm{SR}(\alpha\mathbf{W})\ge\mathrm{SR}(\mathbf{W})$, with strict inequality if any user has nonzero desired power. Consequently the optimum of \eqref{eq:sumrate_obj_modal} is attained on the sphere $\|\mathbf{W}\|_{\mathrm F}^{2}=P$, and the projection $\mathbf{W}\mapsto\sqrt{P}\,\mathbf{W}/\|\mathbf{W}\|_{\mathrm F}$ never decreases the sum rate.
\end{lemma}
\begin{proof}
Writing $S_n=|\boldsymbol{\gamma}_n^{\top}\mathbf{w}_n|^2$ and $I_n=\sum_{k\neq n}|\boldsymbol{\gamma}_n^{\top}\mathbf{w}_k|^2$, scaling gives $\mathrm{SINR}_n(\alpha)=\alpha^2 S_n/(\alpha^2 I_n+\sigma_n^2)$, whose derivative in $\alpha^2$ is $S_n\sigma_n^2/(\alpha^2 I_n+\sigma_n^2)^2\ge0$: signal and interference scale together while the noise does not.
\end{proof}
Accordingly, the beamforming block consists of the exact update \eqref{eq:W_closedform_deriv} followed by the radial projection of Lemma~\ref{lem:radial}, and the receive scalars and weights are refreshed afterwards so that the \ac{WMMSE} bound is tight at the projected point.

\subsection{Modal Matrix Update on the Unitary Manifold}
\label{sec:modal_update}
Fixing $\{\mathbf{w}_k\}$, $\{v_n\}$, and $\{q_n\}$, the dependence of \eqref{eq:wmmse_equiv} on $\boldsymbol{\Psi}$ enters through \eqref{eq:g_eff_def}, where only the first term of $\boldsymbol{\gamma}_n$ varies with $\boldsymbol{\Psi}$ while the offset $\boldsymbol{\gamma}_n^{0}$ contributes the constant $d_{n,k}=(\boldsymbol{\gamma}_n^{0})^{\top}\mathbf{w}_k$ to each gain $h_{n,k}=\boldsymbol{\beta}_n^{\top}\boldsymbol{\Psi}\boldsymbol{\alpha}_k+d_{n,k}$. Dropping $\boldsymbol{\Psi}$-independent terms, the objective is the smooth function
\begin{align}
\mathcal{L}(\boldsymbol{\Psi})
&=\sum_{n=1}^{K}q_n\Bigl(|v_n|^2\sum_{k=1}^{K}\bigl|\boldsymbol{\beta}_n^{\top}\boldsymbol{\Psi}\boldsymbol{\alpha}_k+d_{n,k}\bigr|^2\nonumber\\
&\qquad\qquad-2\Re\{v_n^{*}(\boldsymbol{\beta}_n^{\top}\boldsymbol{\Psi}\boldsymbol{\alpha}_n+d_{n,n})\}\Bigr)\nonumber\\
&=\mathrm{Tr}\bigl(\boldsymbol{\Psi}^{\mathsf{H}}\mathbf{Y}\boldsymbol{\Psi}\mathbf{Z}\bigr)-2\Re\bigl\{\mathrm{Tr}\bigl(\boldsymbol{\Psi}^{\mathsf{H}}\mathbf{X}\bigr)\bigr\},
\label{eq:L_compact_hermitian}
\end{align}
where the constant $d_{n,k}$ enters only the linear block, whose cross terms with $\boldsymbol{\beta}_n^{\top}\boldsymbol{\Psi}\boldsymbol{\alpha}_k$ combine with the desired-signal term into
\begin{align}
\mathbf{X}&=\sum_{n=1}^{K}q_n\,\boldsymbol{\beta}_n^{*}\Bigl(v_n\,\boldsymbol{\alpha}_n^{\mathsf{H}}-|v_n|^2\sum_{k=1}^{K}d_{n,k}\,\boldsymbol{\alpha}_k^{\mathsf{H}}\Bigr)\in\mathbb{C}^{L\times L}, \label{eq:X_modal_hermitian}\\
\mathbf{Y}&=\sum_{n=1}^{K}q_n|v_n|^2\,\boldsymbol{\beta}_n^{*}\boldsymbol{\beta}_n^{\top}\in\mathbb{C}^{L\times L}, \label{eq:Y_modal_hermitian}\\
\mathbf{Z}&=\sum_{k=1}^{K}\boldsymbol{\alpha}_k\boldsymbol{\alpha}_k^{\mathsf{H}}\in\mathbb{C}^{L\times L},\label{eq:Z_modal_hermitian}
\end{align}
with $\mathbf{Y},\mathbf{Z}$ Hermitian and positive semidefinite. Setting $\boldsymbol{\gamma}_n^{0}=\mathbf{0}$ (lossy model) reduces $\mathbf{X}$ to $\sum_n q_n v_n\boldsymbol{\beta}_n^{*}\boldsymbol{\alpha}_n^{\mathsf{H}}$. Crucially, all three blocks are of rank at most $K\ll L$ and admit the factored forms
\begin{equation}
\mathbf{Y}=\mathbf{Y}_{\mathrm f}\mathbf{Y}_{\mathrm f}^{\mathsf H},\quad
\mathbf{Z}=\mathbf{A}_{\mathrm f}\mathbf{A}_{\mathrm f}^{\mathsf H},\quad
\mathbf{X}=\mathbf{X}_{\mathrm f}\mathbf{A}_{\mathrm f}^{\mathsf H},
\label{eq:lowrank_factors}
\end{equation}
with $\mathbf{Y}_{\mathrm f},\mathbf{A}_{\mathrm f},\mathbf{X}_{\mathrm f}\in\mathbb{C}^{L\times K}$ ($\mathbf{A}_{\mathrm f}=[\boldsymbol{\alpha}_1,\dots,\boldsymbol{\alpha}_K]$, $\mathbf{Y}_{\mathrm f}$ the $\sqrt{q_n}|v_n|$-scaled $\boldsymbol{\beta}_n^{*}$, and $\mathbf{X}_{\mathrm f}$ collecting both linear-block contributions). This structure is exploited below to make the line search essentially free. The Euclidean (conjugate) gradient of \eqref{eq:L_compact_hermitian} is
\begin{align}
\boldsymbol{\Omega}=\nabla_{\boldsymbol{\Psi}}\mathcal{L}&=\mathbf{Y}\boldsymbol{\Psi}\mathbf{Z}-\mathbf{X}\nonumber\\
&=\bigl(\mathbf{Y}_{\mathrm f}(\mathbf{Y}_{\mathrm f}^{\mathsf H}\boldsymbol{\Psi}\mathbf{A}_{\mathrm f})-\mathbf{X}_{\mathrm f}\bigr)\mathbf{A}_{\mathrm f}^{\mathsf H},
\label{eq:eucl_grad}
\end{align}
computable in $O(L^{2}K)$. Because \eqref{eq:psi_unitary_modal} confines $\boldsymbol{\Psi}$ to the unitary group $\mathcal{U}(L)$, we follow the unitary steepest-descent method of \cite{abrudan2008steepest}. The Riemannian gradient is the skew-Hermitian matrix
\begin{equation}
\boldsymbol{\Xi}=\boldsymbol{\Omega}\boldsymbol{\Psi}^{\mathsf{H}}-\boldsymbol{\Psi}\boldsymbol{\Omega}^{\mathsf{H}},\qquad \boldsymbol{\Xi}^{\mathsf{H}}=-\boldsymbol{\Xi},
\label{eq:riem_grad}
\end{equation}
and the update moves along the corresponding geodesic,
\begin{equation}
\boldsymbol{\Psi}^{(i+1)}=\exp\!\bigl(-\eta\,\boldsymbol{\Xi}\bigr)\,\boldsymbol{\Psi}^{(i)} .
\label{eq:phi_update}
\end{equation}
Since $\boldsymbol{\Xi}$ is skew-Hermitian, $\exp(-\eta\boldsymbol{\Xi})$ is unitary and \eqref{eq:phi_update} keeps $\boldsymbol{\Psi}^{(i+1)}$ on $\mathcal{U}(L)$ exactly, for any step $\eta>0$. The step $\eta$ is chosen by Armijo backtracking, accepting the first $\eta$ with $\mathcal{L}(\boldsymbol{\Psi}^{(i+1)})\le\mathcal{L}(\boldsymbol{\Psi}^{(i)})-c\,\eta\,\|\boldsymbol{\Xi}\|_{\mathrm F}^{2}$, which ensures a monotone, sufficient decrease of $\mathcal{L}$.

The line search admits a fast implementation. Writing $j\boldsymbol{\Xi}=\mathbf{V}\mathrm{diag}(\mathbf{d})\mathbf{V}^{\mathsf H}$ (one Hermitian eigendecomposition per outer iteration), the geodesic is $\boldsymbol{\Psi}(\eta)=\mathbf{V}\mathrm{diag}(e^{j\eta\mathbf{d}})\mathbf{V}^{\mathsf H}\boldsymbol{\Psi}^{(i)}$, and by \eqref{eq:lowrank_factors}
\begin{align}
\mathcal{L}\bigl(\boldsymbol{\Psi}(\eta)\bigr)
&=\bigl\|\mathbf{P}\,\mathrm{diag}(e^{j\eta\mathbf{d}})\,\mathbf{Q}\bigr\|_{\mathrm F}^{2}-2\Re\Bigl\{\sum_{m=1}^{L}e^{-j\eta d_m}\,\omega_m\Bigr\},
\label{eq:cheap_linesearch}
\end{align}
where $\mathbf{P}=\mathbf{Y}_{\mathrm f}^{\mathsf H}\mathbf{V}\in\mathbb{C}^{K\times L}$, $\mathbf{Q}=\mathbf{V}^{\mathsf H}\boldsymbol{\Psi}^{(i)}\mathbf{A}_{\mathrm f}\in\mathbb{C}^{L\times K}$, and $\boldsymbol{\omega}$ collects the diagonal of $(\mathbf{V}^{\mathsf H}\mathbf{X}_{\mathrm f})(\mathbf{A}_{\mathrm f}^{\mathsf H}(\boldsymbol{\Psi}^{(i)})^{\mathsf H}\mathbf{V})$, all precomputed once per outer iteration. Each backtracking trial then costs only $O(K^{2}L)$, instead of the $O(L^{3})$ of a naive evaluation, and the full matrix $\boldsymbol{\Psi}(\eta)$ is formed a single time, after acceptance.

\subsection{Overall Algorithm and Convergence}
Algorithm~\ref{alg:sumrate_modal} alternates the two blocks. For fixed $\boldsymbol{\Psi}$ it forms the effective channels $\{\boldsymbol{\gamma}_n\}$ via \eqref{eq:g_eff_def}, updates $\{v_n,q_n\}$ in closed form, computes $\{\mathbf{w}_k\}$ via \eqref{eq:W_closedform_deriv}, projects onto the power sphere, and refreshes $\{v_n,q_n\}$. For fixed $\{\mathbf{w}_k\}$ it forms the factors \eqref{eq:lowrank_factors} and takes a backtracked Riemannian step via \eqref{eq:eucl_grad} to \eqref{eq:cheap_linesearch}. This composition is monotone not only in the surrogate but in the original objective.

\begin{proposition}[Sum-rate monotonicity]
\label{prop:monotone}
Each full iteration of Algorithm~\ref{alg:sumrate_modal} does not decrease the sum rate \eqref{eq:sumrate_def}. Being bounded above on the compact feasible set, the sum-rate sequence converges.
\end{proposition}
\begin{proof}
At the closed-form optima \eqref{eq:mmse_receiver} and $q_n^{\star}=1+\mathrm{SINR}_n$, the surrogate is tight: $\mathcal{L}=K-\ln2\cdot\mathrm{SR}$. For arbitrary $\{v_n,q_n\}$, $\mathcal{L}\ge K-\ln2\cdot\mathrm{SR}$ since the optima minimize each $\mathrm{WMSE}_n$. Hence any step that does not increase $\mathcal{L}$ at fixed $\{v_n,q_n\}$ satisfies $K-\ln2\cdot\mathrm{SR}_{\mathrm{new}}\le\mathcal{L}_{\mathrm{new}}\le\mathcal{L}_{\mathrm{old}}=K-\ln2\cdot\mathrm{SR}_{\mathrm{old}}$, i.e., does not decrease the rate. The beamformer update \eqref{eq:W_closedform_deriv} is the exact constrained minimizer of \eqref{eq:LW_quadratic}, and the Armijo-accepted manifold step never increases \eqref{eq:L_compact_hermitian}. Both are of this form, the latter after the post-projection refresh of $\{v_n,q_n\}$ restores tightness. The radial projection increases the rate directly by Lemma~\ref{lem:radial}. Every step in the composition is therefore rate-nondecreasing.
\end{proof}

\begin{algorithm}[t]
\caption{AO Sum-Rate Maximization for Modal XL-BD-RIS}
\label{alg:sumrate_modal}
\begin{algorithmic}[1]
\STATE \textbf{Input:} $\mathbf{G}$, $\{\mathbf{f}_n\}$, $\boldsymbol{\Phi}$, power $P$, tolerances $\varepsilon,\xi$, backtracking factor $\tau\in(0,1)$, Armijo constant $c$
\STATE \textbf{Initialize:} $\boldsymbol{\Psi}\in\mathcal{U}(L)$, $\{\mathbf{w}_k\}$ feasible; precompute $\boldsymbol{\beta}_n\gets\boldsymbol{\Phi}^{\top}\mathbf{f}_n$, $\mathbf{G}^{\top}\boldsymbol{\Phi}^{*}$, $\boldsymbol{\gamma}_n^{0}\gets\mathbf{G}^{\top}\mathbf{f}_n-(\mathbf{G}^{\top}\boldsymbol{\Phi}^{*})\boldsymbol{\beta}_n$
\REPEAT
    \STATE // WMMSE beamforming block
    \STATE $\boldsymbol{\gamma}_n\gets\mathbf{G}^{\top}\boldsymbol{\Phi}^{*}\boldsymbol{\Psi}^{\top}\boldsymbol{\beta}_n+\boldsymbol{\gamma}_n^{0}$; $\ v_n$ via \eqref{eq:mmse_receiver}, $q_n=1/\mathrm{MSE}_n$
    \STATE $\mathbf{A}\gets\sum_n q_n|v_n|^2\boldsymbol{\gamma}_n^{*}\boldsymbol{\gamma}_n^{\top}$, $\ \mathbf{B}\gets\sum_n q_n v_n\boldsymbol{\gamma}_n^{*}\boldsymbol{\epsilon}_n^{\mathsf{H}}$
    \STATE $\mathbf{W}\gets(\mathbf{A}+\mu\mathbf{I})^{-1}\mathbf{B}$, $\mu\ge0$ by bisection s.t. \eqref{eq:power_constraint_modal}
    \STATE $\mathbf{W}\gets\sqrt{P}\,\mathbf{W}/\|\mathbf{W}\|_{\mathrm F}$ (radial projection, Lemma~\ref{lem:radial}); refresh $v_n,q_n$
    \STATE // Modal RIS block
    \STATE $\boldsymbol{\alpha}_k\gets\boldsymbol{\Phi}^{\mathsf H}\mathbf{G}\mathbf{w}_k\ \forall k$, $\ d_{n,k}\gets(\boldsymbol{\gamma}_n^{0})^{\top}\mathbf{w}_k$; form $\mathbf{Y}_{\mathrm f},\mathbf{A}_{\mathrm f},\mathbf{X}_{\mathrm f}$ via \eqref{eq:lowrank_factors}
    \STATE $\boldsymbol{\Omega}\gets(\mathbf{Y}_{\mathrm f}(\mathbf{Y}_{\mathrm f}^{\mathsf H}\boldsymbol{\Psi}\mathbf{A}_{\mathrm f})-\mathbf{X}_{\mathrm f})\mathbf{A}_{\mathrm f}^{\mathsf H}$; $\ \boldsymbol{\Xi}\gets\boldsymbol{\Omega}\boldsymbol{\Psi}^{\mathsf{H}}-\boldsymbol{\Psi}\boldsymbol{\Omega}^{\mathsf{H}}$
    \STATE Eigendecompose $j\boldsymbol{\Xi}=\mathbf{V}\mathrm{diag}(\mathbf{d})\mathbf{V}^{\mathsf H}$; precompute $\mathbf{P},\mathbf{Q},\boldsymbol{\omega}$ of \eqref{eq:cheap_linesearch}; $\eta\gets\eta_0$
    \REPEAT
        \STATE evaluate $\mathcal{L}(\boldsymbol{\Psi}(\eta))$ via \eqref{eq:cheap_linesearch}
        \IF{$\mathcal{L}(\boldsymbol{\Psi}(\eta))\le\mathcal{L}(\boldsymbol{\Psi})-c\,\eta\,\|\boldsymbol{\Xi}\|_{\mathrm F}^2$}
            \STATE $\boldsymbol{\Psi}\gets\mathbf{V}\mathrm{diag}(e^{j\eta\mathbf{d}})\mathbf{V}^{\mathsf H}\boldsymbol{\Psi}$; \textbf{break}
        \ELSE
            \STATE $\eta\gets\tau\,\eta$
        \ENDIF
    \UNTIL{$\eta<\varepsilon$}
\UNTIL{sum-rate change $<\xi$}
\STATE \textbf{Return:} $\{\mathbf{w}_k\}$, $\boldsymbol{\Psi}$
\end{algorithmic}
\end{algorithm}

\subsection{Computational Complexity}
The basis-dependent products $\mathbf{G}^{\top}\boldsymbol{\Phi}^{*}$, $\boldsymbol{\Phi}^{\mathsf H}\mathbf{G}$, $\{\boldsymbol{\beta}_n\}$, and $\{\boldsymbol{\gamma}_n^{0}\}$ cost $O(NML)$ but are computed once, since $\boldsymbol{\Phi}$ is fixed throughout the AO. Per outer iteration, the WMMSE block forms $\{\boldsymbol{\gamma}_n\}$ in $O(KL^2+K^2L)$, builds $\mathbf{A},\mathbf{B}$ in $O(N^2K)$, and solves \eqref{eq:W_closedform_deriv} through one $N\times N$ eigendecomposition plus $O(N)$ per bisection step. The modal block assembles the factors \eqref{eq:lowrank_factors} in $O(LK^2)$, computes the gradient and skew-Hermitian generator in $O(L^{2}K)$, performs one $L\times L$ Hermitian eigendecomposition in $O(L^3)$, runs $I_\eta$ backtracking trials at $O(K^{2}L)$ each, and reconstructs the accepted $\boldsymbol{\Psi}$ in $O(L^3)$. Aggregating over $I_{\mathrm{out}}$ outer iterations gives
\begin{equation}
O\!\bigl(NML+I_{\mathrm{out}}\,(N^3+L^3+I_\eta K^2L)\bigr),
\end{equation}
i.e., the line search no longer multiplies the cubic term. For $L=2\rho\ll M$, guaranteed by Proposition~\ref{prop:halmos} under the geometric model, this is far below the $O(M^3)$ per-iteration cost of optimizing the surface directly in the element domain, the comparison holding for $L\ll M$ and excluding the one-time compression of Sec.~\ref{sec:construction}. The full architecture ladder is quantified in Sec.~\ref{sec:complexity}.

\subsection{Newton-Accelerated Surface Updates}
\label{sec:newton_modal}
The single geodesic step of Sec.~\ref{sec:modal_update} is a first-order update, and its practical convergence is dictated by the conditioning of the quadratic surrogate \eqref{eq:L_compact_hermitian}, whose curvature spans the squared singular-value spread of the atom matrix. In near-field geometry this spread is large, so gradient steps contract slowly along the weakly coupled modes. We therefore solve each surface subproblem to higher accuracy before returning to the beamforming block, replacing the single step by a small number of Riemannian Newton iterations, each Newton system being solved inexactly by truncated conjugate gradients while the subproblem itself is driven to tolerance. In the tangent coordinates of \eqref{eq:phi_update}, i.e., skew-Hermitian directions $\boldsymbol{\Delta}$, the Hessian of the pullback of \eqref{eq:L_compact_hermitian} along the geodesic acts as
\begin{align}
\mathcal{H}[\boldsymbol{\Delta}]&=\mathbf{Y}\boldsymbol{\Delta}\mathbf{B}+\mathbf{B}\boldsymbol{\Delta}\mathbf{Y}-\bigl(\mathbf{C}\boldsymbol{\Delta}+\boldsymbol{\Delta}\mathbf{C}^{\mathsf H}\bigr),\nonumber\\
\mathbf{B}&=\boldsymbol{\Psi}\mathbf{Z}\boldsymbol{\Psi}^{\mathsf H},\qquad
\mathbf{C}=\boldsymbol{\Omega}\boldsymbol{\Psi}^{\mathsf H},
\label{eq:hessian_action}
\end{align}
so that the quadratic and linear blocks of the surrogate enter through the two matrices $\mathbf{B}$ and $\mathbf{C}$ alone, and the Riemannian gradient \eqref{eq:riem_grad} is recovered as their skew part, $\boldsymbol{\Xi}=\mathbf{C}-\mathbf{C}^{\mathsf H}$. Both matrices have rank at most $K$ by \eqref{eq:lowrank_factors}, since $\mathbf{B}=(\boldsymbol{\Psi}\mathbf{A}_{\mathrm f})(\boldsymbol{\Psi}\mathbf{A}_{\mathrm f})^{\mathsf H}$ and $\mathbf{C}=\bigl(\mathbf{Y}_{\mathrm f}(\mathbf{Y}_{\mathrm f}^{\mathsf H}\boldsymbol{\Psi}\mathbf{A}_{\mathrm f})-\mathbf{X}_{\mathrm f}\bigr)(\boldsymbol{\Psi}\mathbf{A}_{\mathrm f})^{\mathsf H}$, and one product \eqref{eq:hessian_action} costs $O(L^{2}K)$. The Newton direction is obtained by truncated conjugate gradients on \eqref{eq:hessian_action}, mapped to the manifold by the geodesic exponential \eqref{eq:phi_update}, and accepted under the same Armijo test, so accepted steps never increase the surrogate and Proposition~\ref{prop:monotone} applies verbatim. The method is globally monotone through the safeguard and inherits the fast local convergence typical of Newton-type schemes on matrix manifolds. Its cost adds a conjugate-gradient budget of at most a small multiple of $L$ products per Newton iteration, so the per-iteration order $O(L^{3}+K^{2}L)$ of Sec.~\ref{sec:complexity} is unchanged, the $L\times L$ eigendecomposition being cubic already. What changes is the iteration count, which drops by more than an order of magnitude.

\section{Element-Domain Architectures: Group-Connected, Diagonal, Fully-Connected}
\label{sec:architectures}
The modal architecture of Secs.~\ref{sec:modal} to \ref{sec:sumrate} is one point in a wider design space. This section adapts the same model and the same optimization machinery to the conventional element-domain family, in which the reconfigurable impedance network is specified directly by the sparsity pattern of $\boldsymbol{\Theta}$, and places all architectures on a common complexity-performance ladder.

\subsection{The Group-Connected Family}
A group-connected surface partitions the $M$ elements into $M/N_g$ groups of size $N_g$ and allows reconfigurable coupling only within groups,
\begin{equation}
\boldsymbol{\Theta}=\mathrm{blkdiag}\bigl(\boldsymbol{\Theta}_1,\dots,\boldsymbol{\Theta}_{M/N_g}\bigr),\qquad \boldsymbol{\Theta}_b\in\mathcal{U}(N_g),
\label{eq:group_def}
\end{equation}
with contiguous index groups unless stated otherwise. The surface carries $M N_g$ complex entries, of which $M(N_g+1)/2$ are free under reciprocity. The family interpolates between the two classical extremes: $N_g=1$ is the single-connected (diagonal) \ac{RIS}, $\boldsymbol{\Theta}=\mathrm{diag}(e^{j\phi_1},\dots,e^{j\phi_M})$, and $N_g=M$ is the fully-connected BD-RIS, $\boldsymbol{\Theta}\in\mathcal{U}(M)$. Because every block-diagonal pattern of size $N_g$ contains every pattern whose size divides $N_g$, the feasible sets are nested along divisor chains, and the achievable rates are monotone in $N_g$ at fixed $M$.

\subsection{Optimization by the Same Alternating Scheme}
The \ac{WMMSE} outer loop of Sec.~\ref{sec:sumrate} is architecture-agnostic: the beamforming block, the receiver and weight updates, Lemma~\ref{lem:radial}, and Proposition~\ref{prop:monotone} depend on the surface only through the effective channels, and apply verbatim with the modal update replaced by an element-domain surface update under the constraint \eqref{eq:group_def}. Writing the surrogate as a function of $\boldsymbol{\Theta}$ (the element-domain image of \eqref{eq:L_compact_hermitian}, with $\boldsymbol{\Phi}=\mathbf{I}_M$, $\mathbf{Y},\mathbf{Z},\mathbf{X}$ built from $\{\mathbf{f}_n^{*}\}$ and $\{\mathbf{e}_k\}$), two update styles cover the whole family.

\subsubsection{Blockwise Riemannian step ($1<N_g\le M$)}
The Euclidean gradient $\boldsymbol{\Omega}_{\boldsymbol{\Theta}}=\mathbf{Y}\boldsymbol{\Theta}\mathbf{Z}-\mathbf{X}$ is projected onto the constraint pattern by retaining its diagonal blocks. Each block then takes the skew-Hermitian geodesic step of Sec.~\ref{sec:modal_update} on $\mathcal{U}(N_g)$, with a shared Armijo backtracking so that the accepted composite step is monotone. The correspondence with the modal update is exact, each block playing the role of $\boldsymbol{\Psi}$ on its own subgroup, and the step is expanded to second order precisely as in Sec.~\ref{sec:newton_modal}. The blockwise Hessian action has the form \eqref{eq:hessian_action} with the block-diagonal projection applied to each product, and it couples the blocks only through $K\times K$ mixtures of the per-block factors, so one product costs $O\bigl(M(N_g+K)K\bigr)$, the order of the gradient itself, with the conjugate-gradient budget scaled to the block size. Proposition~\ref{prop:monotone} applies unchanged, the leading-order ladder \eqref{eq:complexity_ladder} is unaffected, and at $N_g=M$ the construction yields the Newton-accelerated fully-connected anchor used throughout as the performance reference.

\subsubsection{Exact coordinate descent ($N_g=1$)}
For the diagonal surface the surrogate decouples across elements. With $\boldsymbol{\Theta}=\mathrm{diag}(e^{j\boldsymbol{\phi}})$, the dependence of the quadratic surrogate on a single phase $\phi_m$, all others fixed, is of the form
\begin{equation}
\mathcal{L}(\phi_m)=a_m+2\,\Re\bigl\{c_m\,e^{j\phi_m}\bigr\},\qquad
\phi_m^{\star}=\pi-\arg(c_m),
\label{eq:cd_phase}
\end{equation}
where $c_m$ collects the coupling of element $m$ to the residual formed by all other elements. Each coordinate update is therefore exact and closed-form, the sweep is rate-monotone after the receiver/weight refresh, and no manifold machinery is required. This solver is not merely a convenience: in every configuration tested it dominates the generic $N_g=1$ manifold treatment in both speed and attained rate, and its terminal solutions serve as certified feasible points (hence lower bounds) for every $N_g\ge1$ by the nesting above.

\subsection{The Complexity-Performance Ladder}
\label{sec:complexity}
Per outer iteration, with $N$ \ac{BS} antennas, $K$ users, and $\rho=r_G+K$ the active-subspace dimension of \eqref{eq:active_subspace}:
\begin{align}
C_{\mathrm{diag}}&=O\bigl(MK(N+K)\bigr) &&\text{(exact CD sweep)},\nonumber\\
C_{\mathrm{group}}&=O\bigl(M\,[NK+N_gK+N_g^2]\bigr) &&\text{(blockwise step)},\nonumber\\
C_{\mathrm{full}}&=O\bigl(M^3\bigr) &&\text{(dense-}\boldsymbol{\Theta}\text{ step)},\nonumber\\
C_{\mathrm{modal}}&=O\bigl(\rho^{3}\bigr)\ \ [\,+\,O(M\rho^{2})\ \text{once}\,] &&\text{(Sec.~\ref{sec:sumrate} at }L=2\rho\text{)}.
\label{eq:complexity_ladder}
\end{align}
The group cost interpolates its two endpoints ($C_{\mathrm{group}}\to C_{\mathrm{diag}}$ as $N_g\to1$, $\to O(M^3)$ as $N_g\to M$) and grows quadratically in the block size. Any method that touches a dense $\boldsymbol{\Theta}$ pays at least $O(M^2K)$ per iteration merely to read it. The modal architecture is the only entry whose per-iteration cost is independent of $M$: the aperture enters once, through the $O(M\rho^2)$ compression of Sec.~\ref{sec:construction}, and the iterative design then lives entirely in dimension $2\rho\le2(N{+}K)$.

\section{Numerical Results}
\label{sec:numerical}

We consider a square $24\times24$ XL-BD-RIS panel with $M=576$ elements on a $\lambda/2$ grid, of side $D\approx12\lambda$. The rate-versus-entries study of Fig.~\ref{fig:sumratevsentries} also includes a $16\times16$ panel placed in the same surroundings. All lengths are expressed in wavelengths, so the deployment is carrier-agnostic. At $f_c=28$~GHz, for instance, the panel is a $12.9$~cm tile serving indoor links ranging from tens of centimeters to a few meters, consistent with the office scenario of Fig.~\ref{fig:scenario}. Distances are given as fractions of the Rayleigh distance $d_{\mathrm F}=2D^{2}/\lambda$.

The \ac{BS} is a uniform $\lambda/2$ array of $N=8$ antennas. Unless stated otherwise it operates deep in the radiative near field of the panel, at $d_{\mathrm{BS}}=0.16\,d_{\mathrm F}$. The surface and the \ac{BS} are mounted $20\lambda$ above the ground. The $K=3$ users stand on the ground on the opposite side, within a $\pm45^{\circ}$ azimuth sector, at ranges drawn uniformly from $\approx[39\lambda,115\lambda]$. This window lies inside the near field of both panel sizes, so the two panels face the identical user deployment. The user positions are the only randomized quantity in the model. In this reference deployment the measured rank of $\mathbf{G}$ is $r_G=4$ and the user atoms are nearly orthogonal. The active-subspace dimension is therefore $\rho=r_G+K=7$, and the exactness budget of Proposition~\ref{prop:halmos} is $L=2\rho=14\ll M$ (for the $16\times16$ panel, $L=12$). Channels are normalized to unit mean cascaded gain and the \ac{SNR} is $0$~dB. This operating point is conservative for the proposed design, because the weak modes that only large budgets can exploit carry the least rate at low \ac{SNR}. Group-connected surfaces follow \eqref{eq:group_def} with contiguous groups along divisor chains of $M$. All surface subproblems are solved with the Newton-accelerated updates Sec.~\ref{sec:newton_modal} and their blockwise counterpart, and every reported point is the best over multiple independent initializations.

Alongside the element-domain family, we benchmark the modal design against a \ac{DFT} beamspace of the same dimension. Its basis collects the $L$ columns of the two-dimensional $M$-point \ac{DFT} matrix, the far-field beams of the panel, that capture the most energy of the stacked atom matrix $[\mathbf{G},\mathbf{F}^{*}]$, ranked once per deployment and nested in $L$. The beamspace simply replaces the modal basis $\boldsymbol{\Phi}$ in the compression and lossless completion of Sec.~III-C and is optimized by the identical algorithm.

\begin{figure}[t]
\centering
\includegraphics[width=0.8\columnwidth]{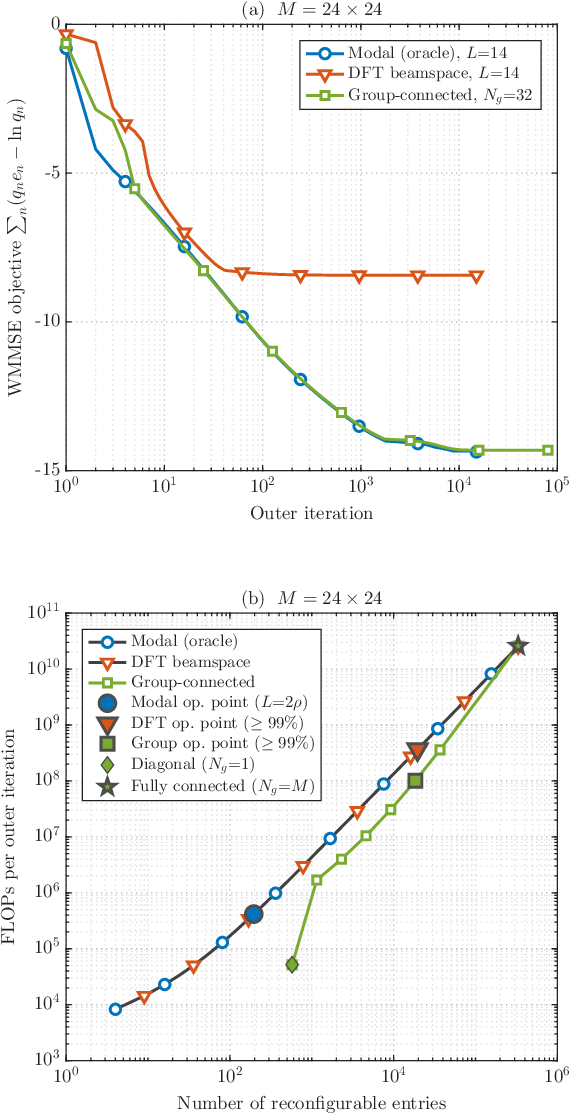}
\caption{Optimization cost at the reference deployment ($M=576$, $0$~dB). (a)~Convergence of the \ac{WMMSE} objective for the modal and \ac{DFT} beamspace designs at $L=2\rho=14$ and a group-connected surface with $N_g=32$. (b)~Leading-order \acp{FLOP} per outer iteration versus reconfigurable entries. Beamforming-stage costs common to all architectures are omitted. Filled markers are operating points: modal at $L=2\rho$ (attaining the fully-connected optimum, Prop.~\ref{prop:halmos}), beamspace and group-connected at their smallest computed budgets reaching $99\%$ of it.}
\label{fig:convcomplex}
\end{figure}
\begin{figure}[t]
\centering
\includegraphics[width=0.8\columnwidth]{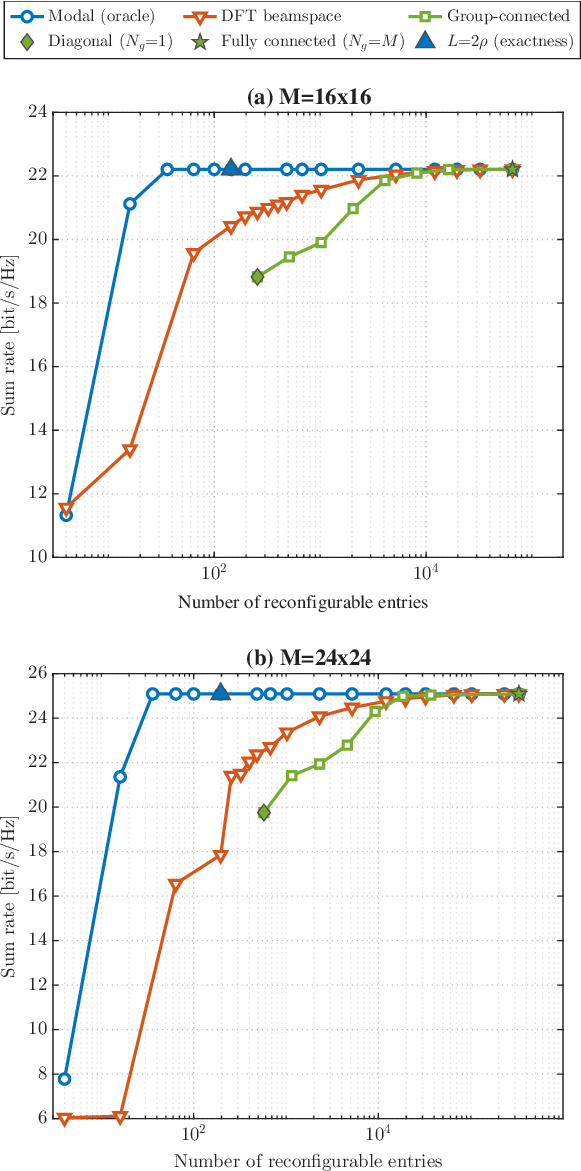}
\caption{Sum rate versus number of reconfigurable entries at $0$~dB for (a)~$M=16\times16$ and (b)~$M=24\times24$ under identical user deployments. Diamond and pentagram mark the diagonal ($N_g=1$) and fully-connected ($N_g=M$) endpoints of the group-connected family.}
\label{fig:sumratevsentries}
\end{figure}
\subsection{Convergence and Complexity}
\label{sec:num_conv}
Fig.~\ref{fig:convcomplex} examines the cost of optimization, deliberately split into its two factors: iterations to converge in panel~(a) and work per iteration in panel~(b). Total complexity is their product, and neither factor alone ranks the architectures. Panel~(a) shows cold-start traces of the \ac{WMMSE} objective $\sum_n(q_ne_n-\ln q_n)$, so the curves are sum-rate trajectories under a monotone transform. The modal design at $L=2\rho=14$ converges in roughly $2\times10^{4}$ outer iterations. The \ac{DFT} beamspace at the same $L$ converges an order of magnitude faster, but to a visibly worse objective. A mismatched basis produces an easier optimization landscape precisely because it cannot reach the better optima, a mechanism quantified in Sec.~\ref{sec:num_dbs}. The group-connected trace ($N_g=32$, cf.\ \eqref{eq:group_def}) is the slowest, as a direct consequence of its structure. Cross-block coordination occurs only through the \ac{WMMSE} re-weighting, so finer partitions take more, cheaper outer steps. Iteration counts alone would therefore mislead, which is the point of panel~(b), where the leading-order \ac{FLOP} counts are drawn against the number of reconfigurable entries. Modal and beamspace designs share one law, since they share the solver and differ only in the fixed basis, and are drawn as a single gray curve with alternating family markers. The group-connected family follows $O(MN_g^{2})$ from the diagonal to the fully-connected endpoint, and all laws meet at $M^{2}$ entries. What separates the families is the operating point, marked filled. The modal design attains the fully-connected optimum at $(2\rho)^{2}=196$ entries. The beamspace first reaches $99\%$ of it at $L=110$, i.e., $12\,100$ entries, and the group-connected family at $N_g=32$, i.e., $18\,432$ entries. 
Each iteration of the \ac{DFT} at $L=110$ would cost two orders of magnitude more than the modal basis at $L=14$. The group-connected family has a comparable overall complexity to the DFT with lower performance. Therefore, the modal architecture wins on both complexity and performance levels.

\subsection{Sum Rate versus Reconfigurable Entries}
\label{sec:num_entries}
Fig.~\ref{fig:sumratevsentries} reports the central comparison: achieved sum rate as a function of the number of reconfigurable entries, for both panel sizes under identical surroundings. Three behaviors organize the figure. First, the modal curve saturates almost immediately. It is within $0.03$~bit of the fully-connected optimum already at $L=6$ to $8$, i.e., $36$ to $64$ entries, and attains it exactly at the marked $L=2\rho$ point (Prop.~\ref{prop:halmos}). The near-saturation below $2\rho$ is itself informative. $L=2\rho$ is what exactness requires, while the rate degrades gracefully under it, because the dimensions dropped first are the weakly coupled modes that carry the least rate at this \ac{SNR}. Second, the beamspace pays what the modal design avoids. At the modal operating budget it lies several bits short, e.g., $17.9$ versus $25.1$~bit at $196$ entries for $M=576$, and it first reaches $99\%$ of the anchor at $12\,100$ entries. This is a roughly $60\times$ entry-count penalty for optimizing in mismatched coordinates, whose physical origin Sec.~\ref{sec:num_dbs} isolates. Third, the element-domain ladder climbs from the diagonal surface ($18.8$ and $19.8$~bit for the two panels) through the group-connected family toward the fully-connected endpoint, confirming the nesting of Sec.~\ref{sec:complexity}, but reaches the anchor only at $\Theta(M^{2})$ entries. As $M$ grows from $256$ to $576$, the fully-connected entry count quintuples, from $65\,536$ to $331\,776$, while the modal operating point moves only from $144$ to $196$ entries. Even this small increase is geometric rather than dimensional. The deployment is shared in absolute terms, so it lies deeper inside the larger panel's near field, and the measured $r_G$ rises from $3$ to $4$. The budget tracks the relative depth of the surroundings, not the element count, which is the sense in which the cost of the proposed design is $M$-independent.

\subsection{Near-Field Range Resolvability}
\label{sec:num_resolve}
Fig.~\ref{fig:resolve} isolates the capability that distinguishes the near field from the far field: separating two users at the same azimuth by range alone. Two users are placed at a common broadside azimuth. The first is anchored at $r_0\in\{45\lambda,90\lambda\}$, the second at $r_0+\Delta r$, and $\Delta r$ is swept. All designs operate at $L=2\rho$. 

Panel~(a) shows the correlation of the two user atoms collapsing over the axial depth-of-focus scale $2r_0^{2}/d_{\mathrm F}$ (dotted). The marker is a characteristic scale rather than an exact landmark: the correlation halves already at half the marked value and nulls in its vicinity. The residual offsets reflect the axial response of a finite aperture with amplitude taper which only approximates the ideal Dirichlet kernel. What the law predicts exactly is the scaling with range. The two anchors' knees sit a factor $\approx4$ apart, making the quadratic range dependence of near-field focusing directly visible. 

Panels~(b) and~(c) show the effect on the sum rate which rises from a single-user floor, since at $\Delta r\to0$ the two channel vectors are colinear and the spatial multiplexing gain collapses to one, up to the two-user level, with the rate knee at roughly half the depth of focus. The modal design transitions first, while the diagonal surface follows with a moderate lag. The \ac{DFT} beamspace requires substantially larger separations deep in the near field ($r_0=45\lambda$) while approaching the modal design toward the far-field edge ($r_0=90\lambda$). Past the peak, the curves decline because the second user recedes, thus increasing pathloss. The ripples at large $\Delta r$ in panel~(b) superpose this differential path loss with the axial sidelobes of the focusing kernel, visible as the wiggles in panel~(a). They appear only for the near anchor because, within the plotted window, only its kernel is traversed past the main lobe. 

\begin{figure}[t]
\centering
\includegraphics[width=0.8\columnwidth]{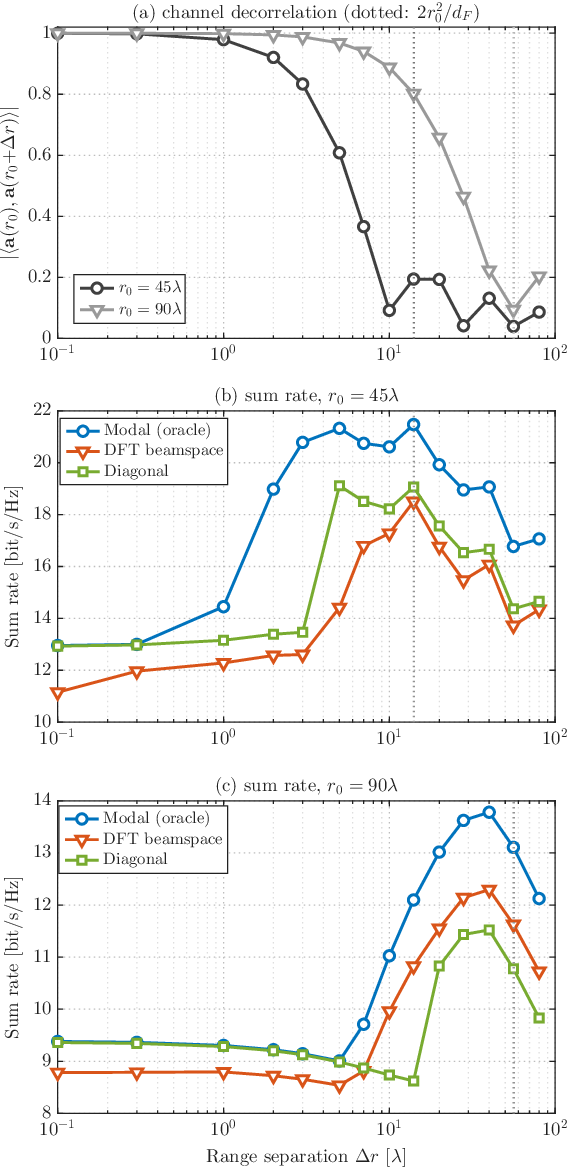}
\caption{Range resolvability of two same-azimuth users ($K=2$, $0$~dB, all restricted designs at $L=2\rho$ with $\rho$ measured per point). (a)~Atom correlation versus range separation for the two anchors. The dotted lines mark the axial depth-of-focus scale $2r_0^{2}/d_{\mathrm F}$. (b), (c)~Sum rate versus separation.} 
\label{fig:resolve} \vspace{-4mm}
\end{figure}

\begin{figure}[t] 
\centering
\includegraphics[width=0.8\columnwidth]{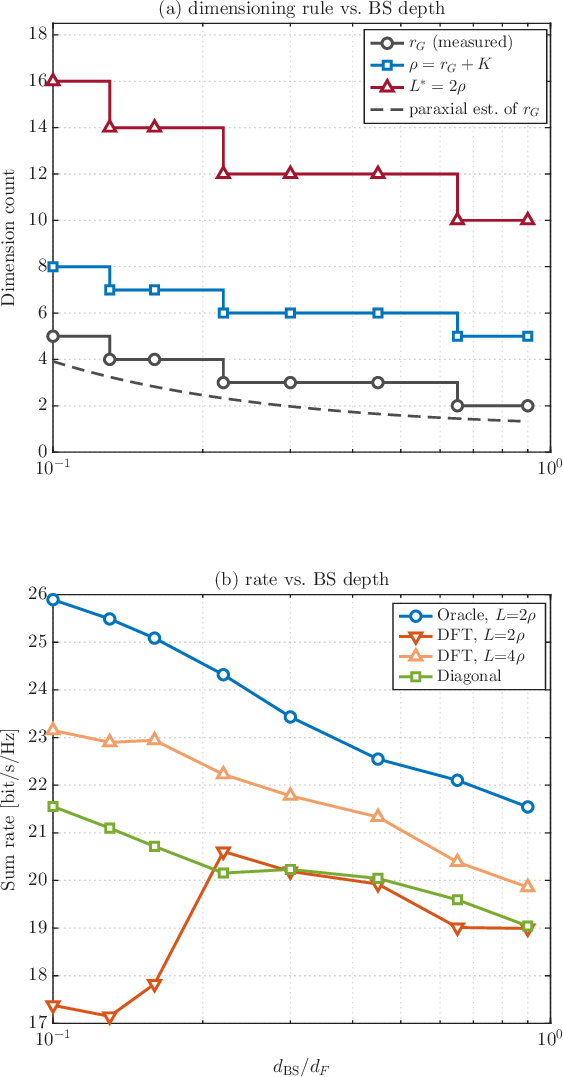}
\caption{Dependence on \ac{BS} depth with the user deployment pinned at $0$~dB. (a)~Measured $r_G$, $\rho=r_G+K$, and the exactness budget $L^{*}=2\rho$ versus $d_{\mathrm{BS}}/d_{\mathrm F}$, with the paraxial estimate of $r_G$ (dashed). (b)~Sum rate versus $d_{\mathrm{BS}}/d_{\mathrm F}$ for the modal design at $L=2\rho$, the \ac{DFT} beamspace at $L=2\rho$ and $L=4\rho$, and the diagonal surface.}
\label{fig:dbs} \vspace{-4mm}
\end{figure}

\subsection{Dependence on Near-Field Depth}
\label{sec:num_dbs}
Fig.~\ref{fig:dbs} closes the study by sweeping the one quantity that moves the dimensioning rule itself, the \ac{BS} depth $d_{\mathrm{BS}}/d_{\mathrm F}\in[0.10,0.90]$, with the users pinned to the identical deployment throughout and the sweep floored above the Fresnel limit. Panel~(a) shows the rule tracking the geometry. The measured $r_G$ steps from $5$ to $2$ as the \ac{BS} recedes, $\rho=r_G+K$ from $8$ to $5$, and the exactness budget $L^{*}=2\rho$ from $16$ to $10$. The dashed paraxial estimate $r_G\approx D_{\mathrm{BS}}D/(\lambda d_{\mathrm{BS}})+1$ tracks the staircase and degrades toward the deep end, and it is the measured rank that the construction of Sec.~\ref{sec:construction} consumes. Panel~(b) reports the consequences at bracketing budgets per basis. The modal design at $L=2\rho$ tracks the fully-connected optimum at every depth, acting as a performance upper-bound. Its halved-budget companion at $L=\rho$ coincides with it across the entire sweep (thus not drawn), demonstrating that the geometric basis can afford half its nominal budget. The \ac{DFT} beamspace cannot afford double. At $L=2\rho$ its curve is non-monotone, rising steeply to a peak near $0.22\,d_{\mathrm F}$ and then declining, and at $L=4\rho$ the same shape survives at reduced amplitude, offset toward the optimum but never reaching it.

The mechanism is a span deficiency with a closed-form signature. A spherical atom is a chirp across the aperture, and its expansion on the beamspace occupies approximately $(d_{\mathrm F}/2r)^{2}$ beams. The beamspace curve therefore decomposes into two regimes. Deep, the budget cannot even carry the incident field. The fraction of \ac{BS}-leg energy expressible by the selected beams is $58\%$ at the deepest point for $L=2\rho$, and its release as the \ac{BS} recedes, from $58\%$ to $98\%$, produces the rising edge. Past the peak the curve simply rides the declining optimum. Doubling the budget to $4\rho$ raises the deep incident-field coverage to $84\%$, which is why its knee nearly vanishes. What no \ac{BS} repositioning cures is the user side. The users' beam spreads are fixed by their own unchanging geometry, capping the beamspace at a persistent offset from the optimum at every depth, consistent with the $L\approx110$ needed in Fig.~\ref{fig:sumratevsentries}. The diagonal surface completes the picture from the architecture side. It possesses the full span but no inter-element mixing, and its gap to the optimum widens with depth, from $2.5$~bit at $0.90\,d_{\mathrm F}$ to $4.4$~bit at $0.10\,d_{\mathrm F}$. A near-rank-one incident field can be steered by per-element phasing alone, while a rank-$5$ one must be mixed, so the value of beyond-diagonal coupling grows with near-field depth.

 \section{Conclusion}
\label{sec:conclusion}
We presented a geometric modal framework for extremely large BD-RIS in the radiative near field. By exploiting the low-dimensional span of position-dependent spherical-wave atoms, a thin-SVD basis with orthonormal slack yields an $L\times L$ unitary representation with $L=2\rho\leq2(N{+}K)$ that is provably equivalent to an unconstrained $M\times M$ unitary surface, independently of panel size and using localization information only. The proposed WMMSE-Riemannian alternating optimization has per-iteration complexity independent of $M$ and provides a unified comparison with element-domain designs. Simulations confirmed exact equivalence across all operating points, achieving the fully-connected optimum with about $200$ reconfigurable entries, approximately $60\times$ fewer than far-field beamspace. The gain of beyond-diagonal coupling over diagonal phasing increases from $2.5$ to $4.4$ bit with near-field depth, while the modal dimension remains small. Although $L=2\rho$ guarantees exactness, performance degrades gracefully for smaller $L$. Future work will address position uncertainty, multipath, electromagnetic boundary effects, and wideband and multi-surface extensions.

\bibliographystyle{IEEEtran}
\bibliography{bibliography}

\end{document}

%% file: acronyms.tex
\acrodef{AO}{Alternating Optimization}
\acrodef{AWGN}{Additive White Gaussian Noise}
\acrodef{BD-RIS}{Beyond-Diagonal Reconfigurable Intelligent Surface}
\acrodef{BS}{Base Station}
\acrodef{DFT}{Discrete Fourier Transform}
\acrodef{FLOP}{Floating-Point Operation}
\acrodefplural{FLOP}[FLOPs]{Floating-Point Operations}
\acrodef{LoS}{Line of Sight}
\acrodef{MMSE}{Minimum Mean Squared Error}
\acrodef{MSE}{Mean Squared Error}
\acrodef{RIS}{Reconfigurable Intelligent Surface}
\acrodefplural{RIS}[RISs]{Reconfigurable Intelligent Surfaces}
\acrodef{SINR}{Signal-to-Interference-plus-Noise Ratio}
\acrodef{SNR}{Signal-to-Noise Ratio}
\acrodef{SVD}{Singular Value Decomposition}
\acrodef{WMMSE}{Weighted Minimum Mean Squared Error}
\acrodef{XL}{Extremely Large}